\documentclass[11pt, a4paper]{article}

\usepackage{amsmath, amsthm, amssymb, url, bbm}
\usepackage[margin=2cm]{geometry}

\usepackage{tikz}
\usetikzlibrary{automata}

\newcommand{\Tran}{\mathrm{Tran}}
\newcommand{\Sing}{\mathrm{Sing}}

\newcommand{\Sym}{\mathrm{Sym}}
\newcommand{\GF}{\mathrm{GF}}

\newcommand{\rk}{\mathrm{rk}}
\newcommand{\id}{\mathrm{id}}

\newcommand{\pr}{\mathrm{pr}}

\theoremstyle{plain}
\newtheorem{conjecture}{Conjecture}
\newtheorem{corollary}{Corollary}
\newtheorem{lemma}{Lemma}

\newtheorem{theorem}{Theorem}
\newtheorem*{claim*}{Claim}

\theoremstyle{definition}
\newtheorem{definition}{Definition}

\newtheorem{example}{Example}
\newtheorem{remark}{Remark}

\newcommand{\espace}{,}

\begin{document}

\title{Complete Simulation of Automata Networks}
\author{
Florian Bridoux\footnote{Aix-Marseille Univ., Toulon Univ., CNRS, LIS, Marseille, France},
Alonso Castillo-Ramirez\footnote{Departamento de Matem\'aticas, Centro Universitario de Ciencias Exactas e Ingenier\'ias, Universidad de Guadalajara, M\'exico. E-mail: \texttt{alonso.castillor@academicos.udg.mx}},
Maximilien Gadouleau\footnote{Department of Computer Science, Durham University, South Road, Durham, DH1 3LE, UK}
}
\maketitle

\begin{abstract}
Consider a finite set $A$ and an integer $n \geq 1$. This paper studies the concept of complete simulation in the context of semigroups of transformations of $A^n$, also known as finite state-homogeneous automata networks. For $m \geq n$, a transformation of $A^m$ is \emph{$n$-complete of size $m$} if it may simulate every transformation of $A^n$ by updating one coordinate (or register) at a time. Using tools from memoryless computation, it is established that there is no $n$-complete transformation of size $n$, but there is such a transformation of size $n+1$. By studying the the time of simulation of various $n$-complete transformations, it is conjectured that the maximal time of simulation of any $n$-complete transformation is at least $2n$. A transformation of $A^m$ is \emph{sequentially $n$-complete of size $m$} if it may sequentially simulate every finite sequence of transformations of $A^n$; in this case, minimal examples and bounds for the size and time of simulation are determined. It is also shown that there is no $n$-complete transformation that updates all the registers in parallel, but that there exists a sequentally $n$-complete transformation that updates all but one register in parallel. This illustrates the strengths and weaknesses of parallel models of computation, such as cellular automata.
\end{abstract}

\section{Introduction}

Memoryless computation (MC) is a modern paradigm for computing any transformation of $A^n$, with $A$ a finite set and $n \geq 2$, by updating one coordinate at a time while using no memory. Its basic idea was developed in \cite{Bu96,Bu04,BGT09,BGT14,BM00,BM04a,BM04b}, and expanded in \cite{CFG14b,CFG14a,GR15}. The seminal example of MC is the famous XOR swap algorithm, which is analogous to the butterfly network, the canonical example of network coding (see \cite{ACLY00}). In the following paragraphs, we shall introduce notation and review the main definitions of MC.

Let $q$ be the cardinality of $A$. Without loss, we usually regard $A$ as the ring $\mathbb{Z}_{q}= \mathbb{Z}/ q \mathbb{Z}$ or, when $q$ is a prime power, the field $\GF(q)$. Since the case when $q = 1$ is trivial, we shall assume $q \ge 2$ henceforth. We refer the coordinates of $A^n$ as \emph{registers} and the elements of $A^n$ as {\em states}. Denote by $e^k \in A^n$ the state with $1$ at its $k$-th register and zero everywhere else, and by $e^0 \in A^n$ the state with zeros in all its registers. For any $a \in A^n$, we denote by $a_i$ the image of $a$ under the $i$-th coordinate projection.

We are interested in studying \emph{transformations} of $A^n$, i.e., functions from $A^n$ to $A^n$. Denote by $\Tran(A^n)$ the set of all transformations of $A^n$, and by $\Sing(A^n)$ and $\Sym(A^n)$ the set of all singular and nonsingular transformations of $A^n$, respectively. The sets $\Tran(A^n)$ and $\Sing(A^n)$, equipped with the composition of transformations $\circ$, form semigroups called the \emph{full transformation semigroup} on $A^n$ and the \emph{singular semigroup} on $A^n$, respectively. The set $\Sym(A^n)$, equipped with $\circ$, forms a group called the \emph{symmetric group} on $A^n$. 

In general, if $Y$ is a subset of a semigroup $S$, let $\langle Y \rangle$ be the smallest subsemigroup of $S$ containing $Y$. Say that $Y$ is a generating set of $S$ when $S = \langle Y \rangle$. In particular, if $S$ is a subsemigroup of $\Tran(A^n)$, and $Y$ is a generating set of $S$, the triple $(A^n, S, Y)$ is referred as a \emph{finite state-homogeneous automata network} (see \cite[p.~200]{DN05}).

Throughout this paper, we adopt the convention of applying functions on the right; then $(x)f$ denotes the image of $x \in A^n$ under $f \in \Tran(A^n)$, and $f \circ g$ (or simply $fg$) denotes the composition of functions $(x)(f \circ g) = ( (x) f) g$. The size of the image of a transformation $f$ is referred as its \emph{rank} and denoted by $\rk(f)$.

We view each transformation of $A^n$ as a tuple of functions $f = (f_1,\ldots,f_n)$, where $f_i : A^n \to A$ is referred to as the $i$-th coordinate function of $f$. In particular, an $i$-th coordinate function is {\em trivial} if it is equal to the $i$-th projection: $(x)f_i = x_i$, for all $x \in A^n$. 

The following is the key definition of memoryless computation.

\begin{definition}[Instruction] \label{def:instruction}
An {\em instruction} of $A^n$ is a transformation $f : A^n \to A^n$ with at most one nontrivial coordinate function. A permutation instruction is an instruction which maps $A^n$ bijectively onto $A^n$.
\end{definition}

The previous definition implies that the identity transformation of $A^n$ is an instruction. We denote the set of instructions of $A^n$ as $\bar{\mathcal{I}}(A^n)$, and the set of permutation instructions as $\mathcal{I}(A^n)$. We shall simply write $\bar{\mathcal{I}}$ and $\mathcal{I}$ when there is no ambiguity. Note that any nontrivial instruction $f \in \bar{\mathcal{I}}$ is uniquely determined by its nontrivial coordinate function $f_i$; hence, in this case, we say that $f$ {\em updates} the $i$-th register, and we shall often denote $f$ by its \emph{update form}:
\[ f : x_i \gets (x)f_i. \]
For instance, if $A = \GF(2)$ and $n=2$, then $\mathcal{I}$ is given by
\[\begin{tabular}{llll}
    $\{ x_1 \gets x_1$,& $x_1 \gets x_1 + 1$,& $x_1 \gets x_1 + x_2$,& $x_1 \gets x_1 + x_2 +1$,\\
     $x_2 \gets x_2$,& $x_2 \gets x_2 + 1$,& $x_2 \gets x_1 + x_2$,& $x_2 \gets x_1 + x_2 +1\}$,
\end{tabular}\]
where the identity may be represented by either $x_1 \gets x_1$ or $x_2 \gets x_2$.

One of the most important features of the instruction sets $\bar{\mathcal{I}}$ and $\mathcal{I}$ is that they are generating sets of $\Tran(A^n)$ and $\Sym(A^n)$, respectively (see \cite{Bu96,GR15}).

\begin{definition}[Program] \label{def:program}
For any $g \in \Tran (A^n)$, a {\em program} of length $\ell$ computing $g$ is a sequence of instructions $h^{(1)},\ldots,h^{(\ell)} \in \bar{\mathcal{I}}$ such that
$$ g = h^{(1)} \circ \ldots \circ h^{(\ell)}. $$
\end{definition}
For reminder, we apply functions from the left to the right. Thus, the image of $x$ by $g$ is obtained by applying on $x$, first $h^{(1)}$, then $h^{(2)}$, $\dots$, and at least $h^{(l)}$. Unless specified otherwise, we assume that every instruction in a program is different from the identity. Moreover, since the set of instructions updating a given register is closed under composition, we may always assume that $h^{(k+1)}$ updates a different register than $h^{(k)}$ for all $k$. In this paper, we shall work with particular subsets of instructions $Y \subseteq \bar{\mathcal{I}}$. Hence, for any transformation $g \in \langle Y \rangle$, we define the {\em procedural complexity of $g$ with respect to $Y$} as the minimum length of a program computing $g$ with instructions from $Y$. The procedural complexity of $g$ with respect to $\bar{\mathcal{I}}$ is simply called the \emph{procedural complexity of $g$}.

\begin{example}
In order to illustrate our notations, let us write the program computing the swap of two variables, i.e. $g : \mathbb{Z}_q^2 \to \mathbb{Z}_q^2$ where $(x_1,x_2)g = (x_2,x_1)$. It is given as follows:
\[ g = h^{(1)} \circ h^{(2)} \circ h^{(3)}, \]
where
\begin{align*}
    h^{(1)} : \quad x_1 &\gets x_1 + x_2 \\
   h^{(2)} :  \quad x_2 &\gets x_1 - x_2 \\
    h^{(3)} : \quad x_1 &\gets x_1 - x_2,
\end{align*}
or, equivalently
\begin{align*}
	(x_1, x_2) h^{(1)} &= (x_1 + x_2, x_2),\\
	(x_1, x_2) h^{(2)} &= (x_1, x_1 - x_2),\\
	(x_1, x_2) h^{(3)} &= (x_1 - x_2, x_2).
\end{align*}

By simple extension of $g$ we have,
	\begin{align*}
	(x_1, x_2) g &= (x_1, x_2) (h^{(1)} \circ h^{(2)} \circ h^{(3)}),\\
	 &= (x_1+x_2, x_2) (h^{(2)} \circ h^{(3)}),\\
	 &= (x_1+x_2, x_1) h^{(3)}\\
	 &= (x_2, x_1).
	\end{align*}
 
\end{example}

This paper is organised as follows. In Section \ref{sec:sim}, we introduce our notion of \emph{simulation}, which is a way of computing a transformation of $A^n$ using $m \geq n$ instructions that may depend on $m-n$ extra registers. We say that a transformation of $A^m$ is \emph{$n$-complete} if the instructions induced by its coordinate functions may simulate any transformation of $A^n$. We show that there is no $n$-complete transformation that uses no extra registers, but that there is one that uses only one extra register. Then, we construct an $n$-complete transformation with maximal time of simulation $2n$, and conjecture that $2n$ is the lower bound for the maximal time of simulation of any $n$-complete transformation. 

 In Section \ref{sec:seq}, we introduce the notion of \emph{sequential simulation}. A transformation of $A^m$ is \emph{sequentally $n$-complete} if it may sequentially simulate any sequence of transformations of $A^n$. We establish that any such transformation requires at least $n$ extra registers, and we construct one with $n+2$ extra registers when $q \geq 3$, and $n+3$ extra registers when $q=2$. Then, we establish lower bounds for the maximal and minimal time of simulation of sequentially $n$-complete transformations, and construct explicit examples that asymptotically tend to these bounds.  
 
Finally, in Section \ref{sec:asy}, we show that there is no complete transformation that updates all the registers in parallel; however, we construct a sequentially $n$-complete transformation that updates all but one register in parallel. The first result shows that some asynchronism is required in order to obtain completeness; conversely, the second result shows that the least amount of asynchronism is enough to obtain completeness.

Simulation on automata networks is a well-studied subject, e.g. see \cite{DN05, E91, AN87, T79,T82,T83,T85,T86,T88}. The emphasis in the majority of these works has been on the structure of the so called \emph{interaction graph} of $f \in \Tran(A^n)$, which is a directed graph on $\{ 1, \dots, n \}$ with an arc from $j$ to $i$ if and only if $f_i$ really depends on $x_j$. On the other hand, in this paper we always allow every interaction and focus on other aspects such as the space and time of simulations.     

Our work differentiates in several aspects from results on completeness in other models of computation. First, we always consider a finite computational space, so well-known models, such as Turing machines, are incomparable. Second, as we allow our registers to be updated asynchronously, our model is more general and flexible than synchronous models like cellular automata. This point is illustrated by the results in Section \ref{sec:asy}, especially in the sequentially $n$-complete transformation that updates all but one register in parallel. Indeed, this transformation only uses asynchronism to reset a counter, i.e. to place the state in a special initial configuration; once this is done, the parallel updates are then sequentially $n$-complete.



\section{Simulation of transformations} \label{sec:sim}

Denote $[n] := \{1, \dots, n\}$. For $m \geq n$, let $\pr_{[n]} : A^m \to A^n$ be the $[n]$-projection of $A^m$ to $A^n$, i.e., $(x_1, \dots, x_m)\pr_{[n]} = (x_1, \dots, x_n)$. This is extended to any $I \subseteq [m]$ in the natural way. We shall simplify notation and write $x_I = (x)\pr_I$. 

For any $f: A^m \to A^m$ and $i \in [m]$, $F^{(i)} : A^m \to A^m$ is the instruction \emph{induced by the coordinate function} $f_i$:
\[ F^{(i)} : \quad x_i \gets (x)f_i. \]
We then consider
\[ S_f := \langle F^{(1)}, \dots, F^{(m)} \rangle \subseteq \Tran(A^m).\]
In order to make notation more concise, for any sequence $\sigma = (\sigma_1, \dots, \sigma_t)$ of coordinates in $[m]$, we denote
\[
	F^\sigma = F^{(\sigma_1, \dots, \sigma_t)} :=  F^{(\sigma_1)} \circ F^{(\sigma_2)} \circ \cdots \circ F^{(\sigma_t)}.
\]
Then $S_f$ is the set of all possible $F^\sigma$.

 \begin{definition}[Simulation]
 Let $m \geq n \geq 1$. We say that $f: A^m \to A^m$ \emph{simulates} $g : A^n \to A^n$ if there exists $h \in S_f$ such that $h \circ \pr_{[n]} = \pr_{[n]} \circ  g$, or equivalently $(x)h_{[n]} = (x_{[n]})g$ for all $x \in A^m$. The \emph{time of simulation}, denoted by $t_f (g)$, is the procedural complexity of $h$ with respect to $\{F^{(1)}, \dots, F^{(m)} \}$. 
 \end{definition}
 
Compare our previous definition of simulation with the definition of \emph{simulation by projection} for finite state-homogeneous automata networks that appears in \cite[p.~208]{DN05}.

\begin{definition}[$n$-Complete]
 Let $m \geq n \geq 1$. A transformation $f: A^m \to A^m$ is called \emph{$n$-complete of size $m$} if it may simulate any transformation in $\Tran(A^n)$. The \emph{time} of $f$ is $\mathbf{t}_f(n) := \max \{ t_f (g) : g\in \Tran(A^n) \}$.
\end{definition} 

We exhibit a simple example of an $n$-complete transformation. This example will also allow us to introduce some concepts and notation used throughout this paper. We begin by constructing a simple, yet powerful tool: a switch. This will allow us to encode bits (or $q$-ary symbols) and as such, to be able to describe anything we want. Note that we cannot use only one register to encode one bit, because we do not know the initial state of that register. Instead, we will use two registers $a$ and $b$, and we let 
\begin{align*}
	(x_a, x_b) f_a &= x_b,\\
	(x_a, x_b) f_b &= x_a + 1.
\end{align*}
(There are several variants to this construction.) In this case, we can say that the switch is on if $x_a \ne x_b$ and the switch is off if $x_a = x_b$. Then the instruction $F^{(b)}$ turns the switch on, while $F^{(a)}$ turns it off.

\begin{example}\label{ex:complete}
For any $n \ge 2$, there is an elementary example of an $n$-complete transformation $f \in \Tran(A^m)$, with size $m = 2n + 2\mathcal{T}$, where $\mathcal{T} = |\Tran(A^n)|$. In order to describe it, we let $[m] = [n] \cup [n]' \cup \{a_1, \dots, a_\mathcal{T}\} \cup \{b_1, \dots, b_\mathcal{T}\}$, where $[n]' = \{ 1', \dots, n' \}$ has cardinality $n$. We also enumerate by $p^1, \dots, p^\mathcal{T}$, all the transformations in $\Tran(A^n)$. Then $f$ is defined as follows:
\begin{alignat*}{2} 
	(x)f_v &= \begin{cases}
		(x_{[n]'}) p^s_v  & \text{if } x_{a_s} \ne x_{b_s} \text{ and } x_{a_r} = x_{b_r} \  \forall 1 \le r \le \mathcal{T}, r \ne s  \\
		x_v & \text{otherwise.}
	\end{cases}  &\qquad& v \in [n]
	\\ \\
	(x)f_{v'} &= x_v, &\qquad& v' \in [n]' 
	\\ \\
	(x)f_{a_s} &= x_{b_s}, &\qquad& 1 \le s \le \mathcal{T} 
	\\ \\
	(x)f_{b_s} &= x_{a_s} + 1, &\qquad& 1 \le s \le \mathcal{T}.
\end{alignat*}

We now show that $f$ is indeed an $n$-complete transformation. Suppose that we want to simulate $p^s$. Then, this may be achieved as follows.
\begin{description}
	\item[Step 1.] Copy the first $n$ registers into $[n]'$: $F^{(1', \dots, n')}$.
	
	\item[Step 2.] Turn all switches off:  $F^{(a_1, \dots, a_\mathcal{T})}$.
	
	\item[Step 3.] Turn the right switch on: $F^{(b_s)}$.
	
	\item[Step 4.] Compute $p^s$: $F^{(1, \dots, n)}$.
\end{description}
Or more concisely, the transformation $h = F^{\sigma}$, where
\[ 
	\sigma = (1', \dots, n', a_1, \dots, a_\mathcal{T}, b_s, 1, \dots, n),
\]
satisfies $(x) h_{[n]} = (x_{[n]}) p^s$.
\end{example}

\begin{figure}
\centering
\begin{tikzpicture}[scale=1.5]
	\node[state] 	(1) at (0,0) {$1$};
	\node[state] 	(2) at (1,0) {$2$};
	\node 			(3) at (2,0) {$...$};
	\node[state] 	(n) at (3,0) {$n$};
	
	\draw (-0.5,1) -- +(4,0) node[below left, font=\small] {$[n]$: First $n$ registers} -- +(4,-1.5) -- +(0,-1.5) -- +(0,0);

	\begin{scope}[yshift = -2cm]	
	\node[state] 	(11) at (0,0) {$1'$};
	\node[state] 	(22) at (1,0) {$2'$};
	\node 			(33) at (2,0) {$...$};
	\node[state] 	(nn) at (3,0) {$n'$};

	\draw (-0.5,-1) -- +(4,0) node[above left, font=\small] {$[n]'$: Copy of $[n]$} -- +(4,1.5) -- +(0,1.5) -- +(0,0);
	\end{scope}

	\path[-latex] 	(1) edge (11)
					(2) edge (22)
					(n) edge (nn);
	
	\begin{scope}[xshift = 6cm]
	\node[state] 	(a1) at (0,0) {$a_1$};
	\node 			(a2) at (1,0) {$...$};
	\node[state] 	(as) at (2,0) {$a_s$};
	\node 			(a3) at (3,0) {$...$};
	\node[state] 	(af) at (4,0) {$a_\mathcal{T}$};
	
	\node[state] 	(b1) at (0,-1) {$b_1$};
	\node 			(b2) at (1,-1) {$...$};
	\node[state] 	(bs) at (2,-1) {$b_s$};
	\node 			(b3) at (3,-1) {$...$};
	\node[state] 	(bf) at (4,-1) {$b_\mathcal{T}$};
	
	\path 	(b1) edge (a1)
			(bs) edge (as)
			(bf) edge (af);

	\draw (-0.5,1) -- +(5,0) node[below left, font=\small] {$S$: Switches} -- +(5,-2.5) -- +(0,-2.5) -- +(0,0);
	\draw (-0.4, 0.4) -- +(0.8,0) -- +(0.8,-1.8) -- +(0,-1.8) -- +(0,0);
	\draw (2-0.4, 0.4) -- +(0.8,0) -- +(0.8,-1.8) -- +(0,-1.8) -- +(0,0);
	\draw (4-0.4, 0.4) -- +(0.8,0) -- +(0.8,-1.8) -- +(0,-1.8) -- +(0,0);
	\end{scope}
	
	\draw[very thick, -latex] (5.5,0) -- (3.5,0);
	\draw[very thick, -latex] (2, -1.5) -- (2, -0.5);
\end{tikzpicture}
\caption{The $n$-complete transformation of Example \ref{ex:complete}}
\end{figure}
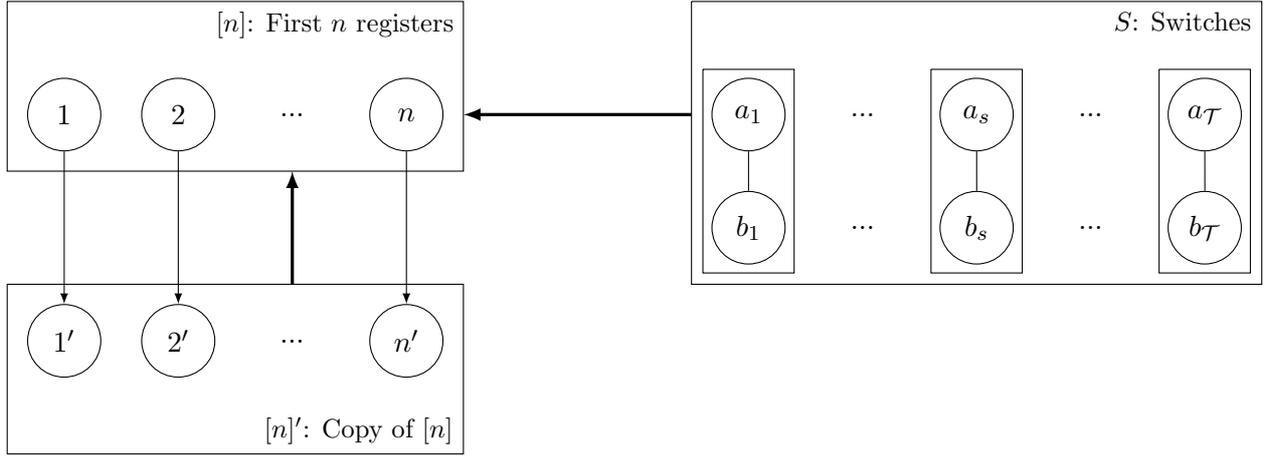

In the following sections, we study $n$-complete transformations with minimal size and time.


\subsection{Complete transformations of minimal size}

In this section, we denote the transposition of $u,v \in A^n$ as $(u,v)$, where, for any $x \in A^n$,
$$
	(x) (u,v) = \begin{cases}
		v &\text{if } x = u\\
		u &\text{if } x = v\\
		x &\text{otherwise,}
	\end{cases}
$$
and the assignment of $u$ to $v$ as $(u \to v)$, where
$$
	(x) (u \to v) = \begin{cases}
		v &\text{if } x = u\\
		x &\text{otherwise.}
	\end{cases}
$$
For any $f \in \Tran(A^n)$ and $g \in \Sym(A^n)$, the conjugation of $f$ by $g$ is $f^g := g^{-1} f g \in \Tran(A^n)$.

It was determined in \cite{CFG14a} that, unless $\lvert A \rvert =n=2$, there exists a set $Y \subset \mathcal{I}$ of size $n$ that generates the whole symmetric group $\Sym(A^n)$; hence, the set $Y \cup \{ (e^0 \to e^1)\}$ of $n+1$ instructions suffices to generate the full transformation semigroup $\Tran(A^n)$. In the following theorem, we prove there is no set of $n$ instructions that generate $\Tran(A^n)$, which implies that there is no $n$-complete transformation of size $n$.

\begin{theorem} \label{th:no_sing}
For any $n \geq 1$, there is no transformation $f \in \Tran(A^n)$ such that $\Sing(A^{n}) \subseteq S_f$.
\end{theorem}

\begin{proof}
The case $n=1$ is trivial, so assume that $n \geq 2$. Suppose that $Y := \{ F^{(1)},\ldots,F^{(n)} \}$ is a set of instructions that generate a semigroup containing all singular transformations, where $F^{(i)}$ updates the $i$-th register. Since the composition of permutations is a permutation, at least one of these generating instructions must be singular. 

First, assume that at least two instructions of $Y$, say $F^{(1)}$ and $F^{(2)}$, are singular. We claim that no assignment $g = (a \to b )$, with $a_i \neq b_i$, $i = 1,2$, can be computed using only instructions in $Y$. Indeed, suppose that $F^{(1)}$ is the first singular instruction in a program computing $g$, so $g = \pi \circ  F^{(1)} \circ h$, for some $h \in \Tran(A^n)$ and $\pi \in \langle F^{(3)},\dots, F^{(n)} \rangle$. As $\pi \circ F^{(1)}$ is singular, there exist $u, v \in A^n$, $u \neq v$, such that $(u) \pi \circ F^{(1)} = (v) \pi \circ F^{(1)}$, which implies that $(u)g = (v) g$. However, as $\pi \circ F^{(1)}$ does not update the second register, we have $\{u,v\} \neq \{ a , b \}$, which contradicts the definition of the assignment $g$.

By the previous paragraph, there may be only one singular instruction in $Y$, say $F^{(1)}$. Let $u,v \in A^n$, $u_1 \neq v_1$, be such that $(u) F^{(1)} = (v) F^{(1)}$. For any $g \in \Sing(A^n)$, we may write $g = \pi \circ F^{(1)}\circ h$, where $h \in \Tran(A^n)$ and $\pi \in \langle F^{(2)}, \dots, F^{(n)} \rangle \subseteq \Sym(A^n)$. Letting $x = (u)\pi^{-1}$ and $y= (v) \pi^{-1}$, we see that $x_1 \neq y_1$ and $(x) g = (y) g$. However, this means that assignments such as $g = ( a \to b )$, with $a \neq b$, $a_1 = b_1$, cannot be computed.
\end{proof}

\begin{corollary}
For any $n \geq 1$, there is no $n$-complete transformation of size $n$.
\end{corollary}
  
In fact, the minimum size is exactly $n+1$.

\begin{theorem}
For all $n \ge 1$ and $q \ge 2$, there exists an $n$-complete transformation of size $n+1$.
\end{theorem}

We first deal with the special case $n=1$.

\begin{lemma}
If $n=1$ and $q=2$, then there exists a $1$-complete transformation of size $2$.
\end{lemma}

\begin{proof}
Let $(x_1, x_2) f = (\neg (x_1 \land x_2), x_1)$. It is easy to verify that $f$ is indeed $1$-complete. We shall do it explicitly in order to illustrate some notation used later on. For all $x = (x_1, x_2)$, we have the following chain, where $y \xrightarrow{i} z$ means $z = (y) F^{(i)}$.
\[
	(x_1, x_2) \xrightarrow{2} (x_1, x_1) \xrightarrow{1} (\neg x_1, x_1) \xrightarrow{1} (1, x_1) \xrightarrow{2} (1,1) \xrightarrow{1} (0,1).
\]
Then all four functions of one Boolean variable (namely, $x_1$, $\neg x_1$, $0$ and $1$) are simulated by $f$.
\end{proof}

\begin{lemma}
If $n=1$ and $q \geq 3$, then there exists a $1$-complete transformation of size $2$.
\end{lemma}
\begin{proof}
Let $f : A^2 \to A^2$ be defined as follows
\begin{align*}
	(x_1, x_2) f_1 & = \begin{cases}
	x_1 + 1 & \text{ if }  x_1 = x_2  \\
	1 & \text{ if }  x_1 = 0   \text{ and } x_2 = q - 1   \\
	0 & \text{ if } x_1 = 1  \text{ and } x_2 = 0    \\
	x_1 & \text{ otherwise }
	\end{cases}  \\
	(x_1, x_2) f_2 &= x_1
\end{align*}
Let us prove that $f$ is $1$-complete.
We shall use the following generating set of $\Tran(A)$: the cycle $c = (0 \espace 1 \dots \espace q-1)$, the transposition $k = (0 \espace 1)$, and the assignment $d = (0 \to 1)$. All that is left to prove is that those transformations, acting on $A$, can be simulated by $f$.
Firstly, it is easy to check that $F^{(2, 1)}$ simulates the cycle $c = (0 \espace 1 \dots \espace q-1)$, since
\[
\begin{array}{lll}
(x_1, x_2) & \xrightarrow{2} (x_1, x_1)& \xrightarrow{1} (x_1+1,x_1)
\end{array}
\]
Secondly, $F^{( (2, 1)^{q},1 )}$ simulates the transposition $k = (0 \espace 1)$, since for any $2 \leq i \leq q-1$,
\[
\begin{array}{lll}
(i, x_2) & \xrightarrow{(2,1)^{q}} (i, i-1)& \xrightarrow{1} (i, i-1)\\
(0, x_2) & \xrightarrow{(2,1)^{q}} (0, q-1)& \xrightarrow{1} (1, q-1)\\
(1, x_2) & \xrightarrow{(2,1)^{q}} (1, 0)& \xrightarrow{1} (0, 0)
\end{array}
\]
Thirdly, $F^{( (2, 1)^{q},1,1 )}$ simulates the assignment $d = (0 \to 1)$, since for any $2 \leq i \leq q-1$,
\[
\begin{array}{llll}
(i, x_2) & \xrightarrow{(2,1)^{q}} (i, i-1)& \xrightarrow{1} (i, i-1) &\xrightarrow{1} (i, i-1)\\
(0, x_2) & \xrightarrow{(2,1)^{q}} (0, q-1)& \xrightarrow{1} (1, q-1) &\xrightarrow{1} (1, q-1)\\
(1, x_2) & \xrightarrow{(2,1)^{q}} (1, 0)& \xrightarrow{1} (0, 0)  &\xrightarrow{1} (1,0)
\end{array}
\]
\end{proof}

We now tackle the case $n = q = 2$.

\begin{lemma}
If $n=2$ and $q=2$, then there exists a $2$-complete transformation of size $3$.
\end{lemma}
\begin{proof}
Let $f$ be defined as
\begin{align*}
(x) f_1 &:= \begin{cases}
x_1 + 1 	& \text{if } x_1 = x_3 \\
x_2 			& \text{otherwise }
\end{cases} \\
(x) f_2 &:= x_2 + x_3\\
(x) f_3 &:= x_1
\end{align*}
Firstly, $F^{(3, 1, 2)}$ simulates the cycle $c = (00 \espace 10 \espace 01 \espace 11)$, since
\[
\begin{array}{llll}
(0, 0, x_3) & \xrightarrow{3} (0, 0, 0)& \xrightarrow{1} (1, 0, 0)& \xrightarrow{2} (1, 0, 0)\\
(0, 1, x_3) & \xrightarrow{3} (0, 1, 0)& \xrightarrow{1} (1, 1, 0)& \xrightarrow{2} (1, 1, 0)\\
(1, 0, x_3) & \xrightarrow{3} (1, 0, 1)& \xrightarrow{1} (0, 0, 1)& \xrightarrow{2} (0, 1, 1)\\
(1, 1, x_3) & \xrightarrow{3} (1, 1, 1)& \xrightarrow{1} (0, 1, 1)& \xrightarrow{2} (0, 0, 1)
\end{array}
\]
Secondly, $F^{(3, 1, 1, 2)}$ simulates the cycle $k = (10 \espace 01 \espace 11)$, since
\[
\begin{array}{lllllll}
(0, 0, x_3) & \xrightarrow{3} (0, 0, 0) & \xrightarrow{1} (1, 0, 0) & \xrightarrow{1} (0, 0, 0) & \xrightarrow{2} (0, 0, 0)\\
(0, 1, x_3) & \xrightarrow{3} (0, 1, 0) & \xrightarrow{1} (1, 1, 0) & \xrightarrow{1} (1, 1, 0) & \xrightarrow{2} (1, 1, 0)\\
(1, 0, x_3) & \xrightarrow{3} (1, 0, 1) & \xrightarrow{1} (0, 0, 1) & \xrightarrow{1} (0, 0, 1) & \xrightarrow{2} (0, 1, 1)\\
(1, 1, x_3) & \xrightarrow{3} (1, 1, 1) & \xrightarrow{1} (0, 1, 1) & \xrightarrow{1} (1, 1, 1) & \xrightarrow{2} (1, 0, 1)
\end{array}
\]
Thirdly, $F^{(3, 1, 1, 2, 1, 2)}$ simulates the transformation $d = (01 \espace 11)(00 \to 10)$, since
\[
\begin{array}{llllll}
(0, 0, x_3) & \xrightarrow{3} (0, 0, 0) & \xrightarrow{1} (1, 0, 0) & \xrightarrow{1} (0, 0, 0) & \xrightarrow{2} (0, 0, 0)  \xrightarrow{1} (1, 0, 0)  \xrightarrow{2} (1, 0, 0) \\
(0, 1, x_3) & \xrightarrow{3} (0, 1, 0) & \xrightarrow{1} (1, 1, 0) & \xrightarrow{1} (1, 1, 0) & \xrightarrow{2} (1, 1, 0)  \xrightarrow{1} (1, 1, 0)  \xrightarrow{2} (1, 1, 0) \\
(1, 0, x_3) & \xrightarrow{3} (1, 0, 1) & \xrightarrow{1} (0, 0, 1) & \xrightarrow{1} (0, 0, 1) & \xrightarrow{2} (0, 1, 1)  \xrightarrow{1} (1, 1, 1)  \xrightarrow{2} (1, 0, 1) \\
(1, 1, x_3) & \xrightarrow{3} (1, 1, 1) & \xrightarrow{1} (0, 1, 1) & \xrightarrow{1} (1, 1, 1) & \xrightarrow{2} (1, 0, 1)  \xrightarrow{1} (0, 0, 1)  \xrightarrow{2} (0, 1, 1)
\end{array}
\]
\end{proof}

We now solve all the other cases.

\begin{lemma}
If $n \ge 3$ and $q=2$ or if $n \ge 2$ and $q \ge 3$, then there exists an $n$-complete transformation of size $n+1$.
\end{lemma}

\begin{proof}
If $n \ge 3$ and $q=2$ or if $n \ge 2$ and $q \ge 3$, by \cite{CFG14a}, there exists a set of $n$ (permutation) instructions $g^{(1)}, \dots, g^{(n)}$ generating $\Sym(A^n)$ such that there exists $z \in A^n$ fixed by $g^{(1)}$ but not by $g^{(2)}$. We then denote the assignment instruction $(x)d = (x)D^{(1)} = (z \to z + e^1)$, where $e^1$ is the unit vector $(1, 0, \dots, 0)$. We also denote the product of the orders of $g^{(v)}$ for $v \in [n]$ as $\Omega$; by definition, $g^{(v)^\Omega} = \id$ for all $v$. Since $d$ is idempotent, we have $d^\Omega = d$. Finally, we denote an element of $A^{n+1}$ as $(x, \alpha)$ where $x \in A^n$ and $\alpha \in A$.

We define $f$ as follows.
\begin{align*}
	(x, \alpha) f_1 &= \begin{cases}
	(x) g_1 &\text{if } \alpha = 0\\
	(x) d_1 &\text{if } \alpha = 1\\
	x_1			& \text{otherwise},
	\end{cases} \\
	\\
	(x, \alpha) f_v &= (x) g_v & 2 \le v \le n\\
	\\
	(x, \alpha) f_{n+1} &= \delta\left( (x, \alpha), (z, 0) \right),
\end{align*}
where $\delta(s,t)$ is the Kronecker delta function.

This time, the initialisation step brings $\alpha$ to $0$. For all $(x, \alpha)$, $(x, \alpha) F^{(n+1, 2, n+1, n+1, (2)^{\Omega - 1})} = (x, 0)$. Indeed, for any $\alpha \in A$, any $\beta \in A \setminus \{0\}$, and any $x \ne z$ we have
\[
\begin{array}{lllll}
	(z, 0) 		& \xrightarrow{n+1} (z, 1)	&\xrightarrow{2} (g^{(2)}(z), 1) &\xrightarrow{(n+1)^2} (g^{(2)}(z), 0) &\xrightarrow{(2)^{\Omega - 1}} (z, 0)\\
	(z, \beta) 	& \xrightarrow{n+1} (z, 0)	&\xrightarrow{2} (g^{(2)}(z), 0) &\xrightarrow{(n+1)^2} (g^{(2)}(z), 0) &\xrightarrow{(2)^{\Omega - 1}} (z, 0)\\
	(x, \alpha) & \xrightarrow{n+1} (x, 0)	&\xrightarrow{2} (g^{(2)}(x), 0) &\xrightarrow{(n+1)^2} (g^{(2)}(x), 0) &\xrightarrow{(2)^{\Omega - 1}} (x, 0).
\end{array}
\]

Thus, we focus on the set $\tilde{A} = \{(x, 0) : x \in A^n \}$ and we prove that $f$ can simulate the generating set $\{ g^{(1)}, \dots,g^{(n)}, d \}$ of $\Tran(A^n)$ acting on $\tilde{A}$. Firstly, $(x, 0) F^{(v)} = ((x) g^{(v)}, 0)$ for all $v \in [n]$. Secondly, $(x, 0) F^{(n+1, (1)^\Omega, n+1)} = ((x)d, 0)$, since for every $y \ne z$ we have
\[
\begin{array}{llll}
	(z, 0) & \xrightarrow{n+1} (z, 1) &\xrightarrow{(1)^\Omega} (z + e^1, 1) 	&\xrightarrow{n+1} (z+ e^1, 0)\\
	(y, 0) & \xrightarrow{n+1} (y, 0) &\xrightarrow{(1)^\Omega} (y, 0) 			&\xrightarrow{n+1} (y, 0).
\end{array}
\]
\end{proof}

\subsection{Time of $n$-complete transformations of size $n+2$}

We now exhibit an $n$-complete transformation of size $n+2$ and time at most $6 \lceil \log_2(q) \rceil (q-1) n q^{n-1} + O(q^{n})$. Before this, we need the following result of memoryless computation.

\begin{theorem} \label{th:MC_n+2}
Let $\lvert A \rvert = q$ and $n \geq 2$. Then $\Tran(A^n)$ is generated by a set of instructions $Y$, containing at most $q$ instructions per register, such that any transformation of $A^n$ has procedural complexity with respect to $Y$ of at most $ 3 \lceil \log_2(q) \rceil (q-1) n q^{n-1} + O(q^{n})$.
\end{theorem}

\begin{proof}
We consider the following instructions:
\begin{align*}
	T^{(1)} : \quad x_1 & \gets x_1 + \delta(x, e^0) - \delta(x, e^1),\\
	A^{(2)} : \quad x_2 & \gets x_2 + \delta(x, e^0),\\
	I^{(1)} : \quad x_1 & \gets x_1 + 1 - \delta(x, e^0) + \delta(x, (q-1)e^1),  \\
	I^{(i)} : \quad x_i & \gets x_i + 1 - \sum_{\lambda \in A} \delta(x, \lambda e^i), && (\text{for } 2 \le i \le n),
\end{align*}
where $\delta(x,y)$ denotes the Kronecker delta function, and $\lambda e^i$ is the state with $\lambda \in A$ in its $i$-th register and zero elsewhere. In order to simplify notation, we shall identify $x \in A^n$ with its lexicographic index $\sum_{i=1}^n x_i q^{i-1} \in \{0, 1, \dots, q^n - 1\}$. With this, we may write $A^{(2)} = (0 \to q)$ and $T^{(1)} = (0, 1)$. Observe that the instructions $I^{(i)}$ are permutations with the following cyclic structure: $I^{(1)}$ consists of one cycle of length $q-1$ and $q^{n-1}-1$ cycles of length $q$, while, for $2 \leq i \leq n$, the instruction $I^{(i)}$ consist of just $q^{n-1}-1$ cycles of length $q$.

Let $\rho :=\lceil \log_2(q) \rceil$ and define 
\[ 
	Y := \left\{ T^{(1)}, \ A^{(2)}, \ (I^{(i)})^{2^{j}} : 1 \le i \le n,  \ 0 \leq j \leq \rho-1 \right\}. 
\]
We shall follow several steps in order to prove that $Y$ is the required generating set.

\begin{enumerate}
\item[(i)] Any transposition $T^{(k)} := (0,k)$, with $k \in A^n$, has procedural complexity with respect to $Y$ of at most $\rho w(k) + O(1)$, where $w(k)$ is the number of non-zero coordinates of $k$.
\begin{proof}
First, we determine the procedural complexity of $(I^{(i)})^\lambda$, for $1 \leq i \leq n$ and $1 \leq \lambda \leq q-1$ with respect to $Y$. Using the binary expansion $\lambda = \sum_{j=1}^\rho \lambda_j 2^{j-1}$, $\lambda_j \in \{0,1\}$, it is clear that
\[  (I^{(i)})^\lambda = (I^{(i)})^{\lambda_1} \circ ((I^{(i)})^2)^{\lambda_1} \circ \dots \circ ((I^{(i)})^{2^{\rho-1}})^{\lambda_\rho}. \]
Thus, we need at most $\rho$ instructions from $Y$ to compute $(I^{(i)})^\lambda$.

Fix $k \in A^n$, and suppose that $1 \leq j_1,\dots,j_w \leq n$, with $w=w(k)$, are the non-zero coordinates of $k$. If $k$ is not a multiple of $q$ (i.e. $j_1 = 1$), we have
\[ T^{(k)}:=  (0,k) = \left( T^{(1)}\right)^{(I^{(1)})^{k_1 - 1} (I^{(j_2)})^{k_{j_2}} \dots (I^{(j_w)})^{k_{j_w}} }, \]
while if $k$ is a multiple of $q$, we have
\[ 	T^{(k)} = \left( T^{(1)} \right)^{(I^{(j_1)})^{k_{j_1}} \dots (I^{(j_w)})^{k_{j_w}} (I^{(1)})^{q-1} }. \] 
The result follows because $\left( (I^{(i)})^{\lambda} \right)^{-1} = (I^{(i)})^{q-\lambda}$, for any $1 \leq  \lambda \leq q-1$ and $2 \leq i \leq n$.
\end{proof}

\item[(ii)]  Any permutation in $\Tran(A^n)$ has procedural complexity with respect to $Y$ of at most $2 \rho (q-1) n q^{n-1} + O(q^{n})$.
\begin{proof}
Note that any transposition $(a,b)$ may be expressed as
\[ 	(a,b) = T^{(b)} T^{(a)} T^{(b)}. \]
Since any permutation with $r$ non-fixed points may be expressed as at most $r - 1$ transpositions, cyclic permutations of length $q^n$ have the maximum procedural complexity. In particular, if $\pi = ( a_1, a_2, \dots, a_{q^n}) \in \Sym(A^n)$, then 
\begin{align*}
	\pi &= (a_1, a_2) \dots (a_{q^n - 1}, a_{q^n})\\
	&= (T^{(a_2)} T^{(a_1)} T^{(a_2)}) (T^{(a_2)} T^{(a_3)} T^{(a_2)}) \dots (T^{(a_{q^n-1})} T^{(a_{q^n-2})} T^{(a_{q^n-1})}) (T^{(a_{q^n-1})} T^{(a_{q^n})} T^{(a_{q^n-1})})\\
	&= (T^{(a_2)} T^{(a_1)} T^{(a_3)} T^{(a_2)}) \dots (T^{(a_{q^n-1})} T^{(a_{q^n-2})} T^{(a_{q^n})} T^{(a_{q^n-1})}).
\end{align*}
In this decomposition, $T^{(a_1)}$ and $T^{(a_{q^n})}$ appear once, while every other transposition $T^{(a_s)}$, $s \not \in \{1, q^n\}$, appears twice. By step (i), $T^{(a_s)}$ requires at most $\rho w(a_s) + O(1)$ instructions from $Y$. Since 
\begin{equation}\label{sum}
\sum_{k \in A^n} w(k) = \sum_{i=1}^{n} i (q-1)^i \binom{n}{i}  = (q-1) n q^{n-1}
\end{equation}
it takes at most
\[ 2 \sum_{s=2}^{q^n-1} \left( \rho w(a_s) + O(1) \right) \le 2 \rho (q-1) n q^{n-1} + O(q^{n}) \]
instructions from $Y$ to compute $\pi$.
\end{proof}

\item[(iii)]  Any transformation in $\Tran(A^n)$ has procedural complexity with respect to $Y$ of at most $3 \rho (q-1) n q^{n-1} + O(q^{n})$.
\begin{proof}
Let $g$ be any transformation of rank $r < q^n$. Consider the partition $\ker(g) := \{ P_1,...,P_r\}$ of $A^n$ induced by the following equivalence relation: $a \sim_{g} b$ if and only if $(a)g = (b)g$. (This equivalence relation is called the \emph{kernel} of $g$). For $1 \leq i \leq r$, let $P_i = \{ p_{i,1}, \dots, p_{i,n_{i}} \}$. Depending on two cases, we shall find a transformation $h$ such that $\ker(g) = \ker(h)$, which implies that $g = h \circ \pi$ for some $\pi \in \Sym(A^n)$. 
\begin{description}
\item[Case 1:] States $0$ and $q$ are in a same set of $\ker(g)$. Without loss, assume $p_{1,1} = 0$ and $p_{1,2} = q$. Then, define
 \[ h := A^{(2)} T^{(p_{1,3})} A^{(2)} \dots T^{(p_{1,n_1})} A^{(2)} (q, p_{2,1}) T^{(p_{2,2})} A^{(2)} \dots T^{(p_{r,n_r})}A^{(2)}. \]

\item[Case 2:] States $0$ and $q$ are in distinct sets of $\ker(g)$. Without loss, assume $p_{1,1}=0$ and $p_{r,n_r} = q$. Then, define
\[ h := (0,q) T^{(p_{1,2})} A^{(2)} \dots T^{( p_{1,n_1})}  A^{(2)}  (q, p_{2,1})  T^{(p_{2,2})} A^{(2)} \dots T^{(p_{r,n_r - 1 })}  A^{(2)}  (p_{j,2},q), \]
where $j$ is the smallest index for which $p_{j,2}$ exists. (Clearly, such an index $j$ always exists because $g$ does not have full rank.) 
\end{description}
Each transposition in $h$ takes at most $ \rho w(p_{i,j}) + O(1)$ instructions and each assignment takes $O(1)$ instructions. The result follows by Equation (\ref{sum}) and Step (ii).
\end{proof}
\end{enumerate}
\end{proof}

\begin{theorem} \label{th:n+2}
There exists an $n$-complete transformation of size $n+2$ and time at most
\[ 6 \lceil \log_2(q) \rceil  (q-1) n q^{n-1} + O(q^{n}). \]
\end{theorem}
\begin{proof}
Let $\rho := \lceil \log_2(q) \rceil $. We consider the generating set of instructions $Y$ given in the proof of Theorem \ref{th:MC_n+2}. For each instruction in $Y$, we denote the corresponding nontrivial coordinate function in lowercase, e.g., the nontrivial coordinate function of $(I^{(1)})^2$ is $(x) i_1^2 = x_1 + 2  - 2\delta(x, e^0) + 2\delta(x, (q-1)e^1)$.

Consider the transformation $f \in \Tran \left( A^{n+2} \right)$ with coordinate functions defined as follows (with $a = n+1$ and $b = n+2$): 
\begin{align*}
	(x) f_1 &= \begin{cases}
		(x_{[n]}) i_1 & \text{if } x_a - x_b = 0 \\
		(x_{[n]}) i_1^{2} & \text{if } x_a - x_b = 1 \\
		\vdots & \vdots \\ 
		(x_{[n]})  i_1^{2^{\rho-1}} & \text{if } x_a - x_b = \rho -1 \\
		(x_{[n]})  t_1 & \text{if } x_a -  x_b= \rho,
	\end{cases} \\ \\
	(x)f_2 &= \begin{cases}
		(x_{[n]}) i_2  & \text{if } x_a - x_b =  0 \\
		(x_{[n]}) i_2^2  & \text{if } x_a - x_b = 1 \\
		\vdots & \vdots \\
		(x_{[n]}) i_2^{2^{\rho-1}}  & \text{if } x_a - x_b = \rho-1 \\
		(x_{[n]}) a_2 & \text{if } x_a - x_b = \rho,
	\end{cases} \\ \\
	(x) f_j &= \begin{cases}
		(x_{[n]}) i_j  & \text{if } x_a - x_b = 0  \\
		(x_{[n]}) i_j^2  & \text{if } x_a - x_b = 1 \\
		\vdots & \vdots \\
		(x_{[n]}) i_j^{2^{\rho-1}}  & \text{if } x_a - x_b = \rho -1 \\
		x_j & \text{if }   x_a - x_b = \rho,
	\end{cases} 
	&& (3 \le j \le n), \\ \\
	(x) f_a		&= x_b,\\ \\
	(x) f_b 	&= x_b + 1.
\end{align*}
The main idea behind the definition of $f$ is that the additional registers $a$ and $b$ work as a switch to decide which instruction the program shall use.

Let $F^{(i)} \in \Tran(A^{n+2})$ be the instruction induced by the coordinate function $f_i$. For any $g \in A^n$, we may now find $h \in S_f$ such that $\pr_{[n]} \circ g = h \circ \pr_{[n]}$. Suppose that $g = g^{(1)} \circ g^{(2)} \circ \dots \circ g^{(\ell)}$, where $g^{(k)} \in Y$. By grouping together the powers of $I^{(j)}$, we may assume that $g^{(k)} \in Y \cup \{ (I^{(j)})^{\lambda} : 1 \leq \lambda \leq q -1 \}$, so $\ell \leq 3 (q-1) n q^{n-1} + O(q^{n})$. Denote $\lambda = \sum_{i=1}^\rho \lambda_i 2^{i-1}$, with $\lambda_i \in \{0,1\}$. Let $h^{(0)} = F^{(a)}$, and for each $1 \leq k \leq \ell$, let
\begin{align*}
	h^{(k)} &= \begin{cases}
		(F^{(b)})^\rho F^{(1)}  F^{(a)}	& \text{if } g^{(k)} = T^{(1)} \\
		(F^{(b)})^\rho F^{(2)}  F^{(a)}	& \text{if } g^{(k)} = A^{(2)} \\
		(F^{(j)})^{\lambda_1} F^{(b)} (F^{(j)})^{\lambda_2} \dots F^{(b)} (F^{(j)})^{\lambda_\rho} F^{(a)} & \text{if } g^{(k)} = (I^{(j)})^{\lambda}.
		\end{cases}
\end{align*}
Therefore, we may take $h = h^{(0)} \circ h^{(1)} \circ  \dots \circ h^{(\ell)}$, which uses at most $2 \rho \ell$ of the instructions $F^{(1)}$, $\dots$, $F^{(n)}$, $F^{(a)}$, $F^{(b)}$.
This shows that $f$ is an $n$-complete transformation of size $n+2$ and time $6 \rho (q-1) n q^{n-1} + O(q^{n})$.
\end{proof}

\begin{remark}\label{rk:1}
For $q =3$ or $q=5$, there is a simpler $n$-complete transformation of size $n+2$ whose time is strictly less than the time of the $n$-complete transformation $f$ constructed in the proof of Theorem~\ref{th:n+2}. Defining 
\begin{align*}
	(x) \tilde{f}_1 &= \begin{cases}
		(x_{[n]})  i_1 & \text{if } x_a = x_b\\
		(x_{[n]})  t_1 & \text{if } x_a \ne x_b,
	\end{cases}\\ \\
	(x)\tilde{f}_2 &= \begin{cases}
		(x_{[n]}) i_2  & \text{if } x_a = x_b\\
		(x_{[n]}) a_2   & \text{if } x_a \ne x_b,
	\end{cases} \\ \\
	(x) \tilde{f}_j &=  (x_{[n]}) i_j, && ( 3 \leq j \leq n) \\ \\
	(x)\tilde{f}_a&= x_b,\\ \\
	(x) \tilde{f}_b &= \begin{cases}
			x_b &\text{if } x_a \ne x_b\\
			x_b + 1 &\text{if } x_a = x_b,
			\end{cases}
\end{align*}
we obtain an $n$-complete transformation $\tilde{f}$ of size $n+2$ and time $\mathbf{t}_{\tilde{f}}(n) = 3 (q-1) n q^{n} + O(q^n)$. 

Observe that, for $q = 7$ or $q \geq 9$, we have $\mathbf{t}_{\tilde{f}}(n) > \mathbf{t}_{f}(n)$, while, for $q \in \{2, 4, 6, 8\}$, $\mathbf{t}_{\tilde{f}}(n) = \mathbf{t}_f(n)$. However, for $q=3$ or $q=5$, we have $\mathbf{t}_{\tilde{f}}(n) < \mathbf{t}_f(n)$; Table \ref{ta:time} compares explicitly the times of $\tilde{f}$ and $f$.
\begin{table}[!h]
\setlength{\tabcolsep}{8pt}
\renewcommand{\arraystretch}{1.6}
\centering
\begin{tabular}{|l|cc|}
\hline
 & $\mathbf{t}_{\tilde{f}}(n)$ & $\mathbf{t}_f(n)$ \\ \hline
 $q=3$ & $6 n 3^{n} + O(3^n)$  &  $8 n 3^{n} + O(3^{n})$   \\
 $q=5$ & $60 n 5^{n-1} + O(5^n)$ & $72 n 5^{n-1} + O(5^{n})$  \\ \hline
\end{tabular}
\caption{Times of $\tilde{f}$ and $f$.} 
\label{ta:time}
\end{table} 
\end{remark}


\subsection{Complete transformations with minimal time}

We now turn to the problem of minimising the time of an $n$-complete transformation. 

\begin{theorem} \label{th:time_2n}
For all $n \ge 1$, there is an $n$-complete transformation of time $2n$.
\end{theorem}

\begin{proof}
Let $Q = q^{q^n}$ and enumerate the functions $A^n \to A$ as $\phi^1, \dots, \phi^Q$. We let $m = n + Qn$ and $[m] = [n] \cup ([n] \times [Q])$. Let $f \in \Tran(A^m)$ be defined as
\begin{alignat*}{2}
	(x)f_v &= \sum_{s = 1}^Q x_{vs} &\qquad& v \in [n]
	\\
	(x)f_{vs} &=  (x_{[n]})\phi^s - \sum_{t \ne s} x_{vt} &\qquad& vs \in [n] \times [Q].
\end{alignat*}
Then let $g = (\phi^{i_1}, \dots, \phi^{i_n}) \in \Tran(A^n)$. It can be simulated as follows.
\begin{description}
	\item[Step 1.] For $v=1$ to $n$, compute the value of $\phi^{i_v}$: $F^{(vi_v)}$.
	
	\item[Step 2.] For $v=1$ to $n$, copy that value into $x_v$: $F^{(v)}$.
\end{description}
\end{proof}

\begin{figure}
\centering
\begin{tikzpicture}[scale=1.5]
\begin{scope}
	\node[state] 	(1) at (0,0) {$1$};
	\node[state] 	(2) at (1,0) {$2$};
	\node 			(3) at (2,0) {$\dots$};
	\node[state] 	(n) at (3,0) {$n$};
	
	\draw (-0.5,1) -- +(4,0) node[below left, font=\small] {$[n]$: First $n$ registers} -- +(4,-1.5) -- +(0,-1.5) -- +(0,0);
\end{scope}

\begin{scope}
	\begin{scope}[xshift = -3.2cm, yshift = -1.2cm]
		\node[state] 	(11) at (0,0) {$11$};
		\node 			(12) at (1*0.8,0) {$\dots$};
		\node[state] 	(1s) at (2*0.8,0) {$1s$};
		\node 			(1t) at (2*0.8,-0.5) {$(x_{[n]})\phi^s - \sum_t x_{(t,1)}$};
		\node 			(13) at (3*0.8,0) {$\dots$};
		\node[state] 	(1q) at (4*0.8,0) {$1Q$};
	\draw (-0.4, 0.4) -- +(4,0) -- +(4,-1.2) -- +(0,-1.2) -- +(0,0);
	\draw[very thick, -latex] (2*0.8,0.4) |- (1);
	\end{scope}
	
	\begin{scope}[xshift = -0.6cm, yshift = -2.7cm]
		\node[state] 	(21) at (0,0) {$21$};
		\node 			(22) at (1*0.8,0) {$\dots$};
		\node[state] 	(2s) at (2*0.8,0) {$2s$};
		\node 			(2t) at (2*0.8,-0.5) {$(x_{[n]})\phi^s - \sum_t x_{(t,2)}$};
		\node 			(23) at (3*0.8,0) {$\dots$};
		\node[state] 	(2q) at (4*0.8,0) {$2Q$};
	\draw (-0.4, 0.4) -- +(4,0) -- +(4,-1.2) -- +(0,-1.2) -- +(0,0);
	\draw[very thick, -latex] (2*0.8,0.4) -- (2);
	\end{scope}

	\begin{scope}[xshift = 3cm, yshift = -1.2cm]
		\node[state] 	(n1) at (0,0) {$n1$};
		\node 			(n2) at (1*0.8,0) {$\dots$};
		\node[state] 	(ns) at (2*0.8,0) {$ns$};
		\node 			(2t) at (2*0.8,-0.5) {$(x_{[n]})\phi^s - \sum_t x_{(t,n)}$};
		\node 			(n3) at (3*0.8,0) {$\dots$};
		\node[state] 	(nq) at (4*0.8,0) {$nQ$};
	\draw (-0.4, 0.4) -- +(4,0) -- +(4,-1.2) -- +(0,-1.2) -- +(0,0);
	\draw[very thick, -latex] (2*0.8,0.4) |- (n);
	\end{scope}
\end{scope}
\end{tikzpicture}
\caption{The $n$-complete transformation of Theorem \ref{th:time_2n}}
\end{figure}
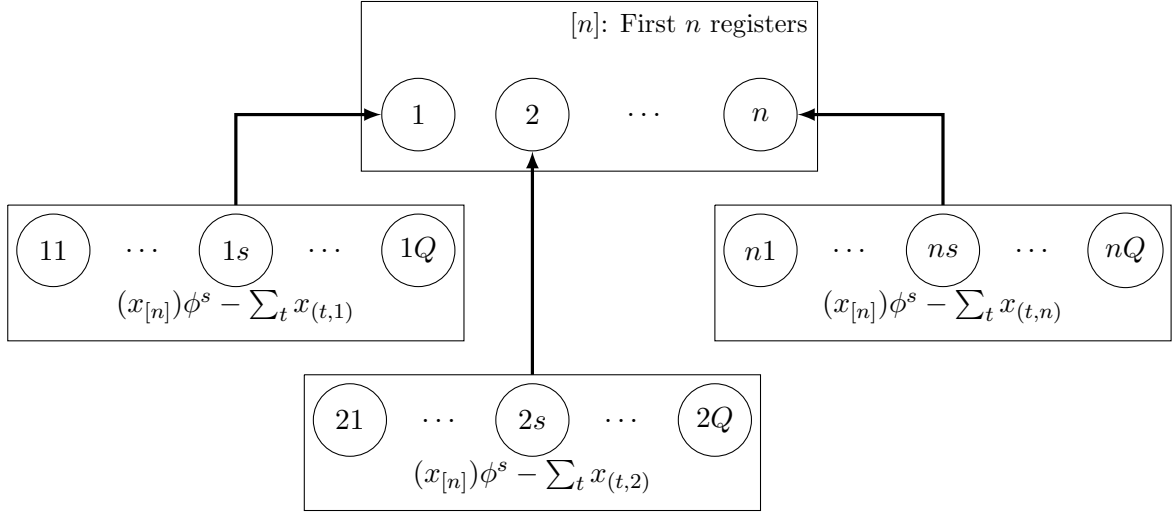

In \cite{B18} it is shown that for any $n$ there exists $h \in \Tran(A^n)$ whose procedural complexity (with respect to $\bar{\mathcal{I}}(A^n)$) is at least $\lfloor 4n/3 \rfloor$, if  $q \geq 2$, and at least $\lfloor 3n/2 - \log_q(n) \rfloor$, if $q \geq 4$. It is clear that the time of simulation $\mathbf{t}_f(n)$ of any $n$-complete transformation $f$ is at least the maximal procedural complexity of any transformation of $\Tran(A^n)$. This gives us some lower bounds for the time of any $n$-complete transformation. In \cite{B18} it is also proved that any transformation $h \in \Tran(A^n)$ can be simulated by a transformation $f \in \Tran(A^m)$ for some $m \geq n$ with a time $\lceil 3n/2 + \log_q(n) \rceil$. However, it is clear that the time to simulate all transformations of $A^n$ is larger than the time to simulate a unique transformation of $A^n$.
	
\begin{conjecture}
The time of any $n$-complete transformation is at least $2n$, for any $n$ and $q$.
\end{conjecture}

\section{Sequential simulation of transformations} \label{sec:seq}

\begin{definition}[Sequential simulation]
Let $m \geq n \geq 1$. We say that $f : A^m \to A^m$ \emph{sequentially simulates} $g^{(1)}, \dots, g^{(\ell)} \in \Tran(A^n)$ if there exist $h^{(1)}, \dots, h^{(\ell)} \in S_f \subseteq \Tran(A^m)$ such that, for any $1 \leq i \leq \ell$,
\[  \pr_{[n]} \circ g^{(i)}  =  h^{(1)} \circ \dots \circ h^{(i)} \circ \pr_{[n]}. \] 
The \emph{sequential time of simulation}, denoted by $\text{st}_f(g^{(1)}, \dots, g^{(\ell)})$, is the procedural complexity of $h^{(1)} \circ \dots \circ h^{(\ell)}$ with respect to $\{F^{(1)}, \dots, F^{(m)}\}$.   
\end{definition}

\begin{definition}[Sequentially $n$-complete]
A transformation of $A^m$ is called \emph{sequentially $n$-complete} if it may sequentially simulate any finite sequence in $\Tran(A^n)$. 
\end{definition}

Define the \emph{maximal and minimal sequential times of $f$}, denoted by $ \mathbf{st}^{\text{max}}_{f}$ and $\mathbf{st}^{\min}_{f}$, respectively, as follows:
\begin{align*}
	\mathbf{st}^{\max}_{f} &= \max \left\{  \text{st}_f(g^{(1)}, \dots, g^{(q^{nq^n})}) : g^{(i)} \neq g^{(j)} \text{ for } i \neq j  \right\}, \\
	\mathbf{st}^{\min}_{f}& = \min \left\{  \text{st}_f(g^{(1)}, \dots, g^{(q^{nq^n})}) :g^{(i)} \neq g^{(j)} \text{ for } i \neq j \right\}.
\end{align*}
 As the sequences considered in the above definitions must include each transformation in $\Tran(A^n)$ exactly once, the relevant aspect when calculating maximal and minimal sequential times is the order in which the transformations of $A^n$ appear in the sequence.   

\begin{example} \label{ex:seq_complete}
The $n$-complete transformation of Example \ref{ex:complete} is in fact a sequentially $n$-complete transformation. Let $g^1 = p^{s_1}, \dots, g^\ell = p^{s_\ell}$, then this sequence can be simulated as follows.
\begin{description}
	\item[Step 1.] Make a copy of the first $n$ registers: $F^{(1', \dots, n')}$.
	
	\item[Step 2.] Turn all switches off: $F^{(a_1, \dots, a_\mathcal{T})}$.
	
	\item[Step 3.] For $j$ from $1$ to $\ell$ do:
	\begin{description}
		\item[Step 3.1.] Turn the switch on: $F^{(b_{s_j})}$.
		
		\item[Step 3.2.] Compute $p^{s_j}$: $F^{(1, \dots, n)}$.
		
		\item[Step 3.3.] Turn the switch off: $F^{(a_{s_j})}$.
	\end{description}
\end{description}
The maximal sequential time of $f$ then satisfies
$$
	\mathbf{st}^{\max}_{f} \le (n + 3 + o(1)) \mathcal{T}.
$$
By considering a sequence where $g^{j-1}$ and $g^j$ only differ by one coordinate function (i.e. a Gray code, which we shall use later on), Step 3.2 can be simplified by updating only the register corresponding to the differing coordinate function. The minimal sequential time of $f$ then satisfies
$$
	\mathbf{st}^{\min}_{f} \le  (4 + o(1))\mathcal{T}.
$$
\end{example}
 
\subsection{Sequentially complete transformations of minimal size}

\begin{theorem} \label{th:space_2n}
The size of a sequentially $n$-complete transformation is at least $2n$.
\end{theorem}

\begin{proof}
Let $f$ be a sequentially $n$-complete transformation of size $m < 2n$. Consider the sequence $g^1, g^2$ of transformations in $\Tran(A^n)$, where $(x)g^1 = 0$ and $(x)g^2 = x$ for all $x \in A^n$. Let $h^1$ simulate $g^1$ and $h^2 \circ h^1$ simulate $g^2$. We have
\[
	\rk(h^1) \ge \rk(h^2 \circ h^1) \ge \rk(g^2) = q^n.
\]
However, since $h^1$ simulates $g^1$, we have
\[
	\rk(h^1) \le q^{m-n} \rk(h^1_{[n]}) < q^n \rk(h^1_{[n]}) = q^n \rk(g^1) = q^n,
\]
which is the desired contradiction.
\end{proof}

%

The exact minimum size of a sequentially $n$-complete transformation is an open problem. Instead, we study the time of a sequentially $n$-complete transformation of size $2n+2$ or $2n+3$.

\begin{theorem} \label{th:space_2n+2}
Let $\lvert A \rvert = q \geq 2$. Then, there exists a sequentially $n$-complete transformation $\hat{f}$ of size $m$ and $\mathbf{st}^{\max}_{\hat{f}}(n)$ as given in Table \ref{ta:sequentially complete}.

\begin{table}[!h]
\setlength{\tabcolsep}{8pt}
\renewcommand{\arraystretch}{1.6}
\centering
\begin{tabular}{|l|cc|}
\hline
 & $m$ & $\mathbf{st}^{\max}_{\hat{f}}(n)$ \\ \hline
 $q=2$ & $2n+3$ & $3 n 2^{n2^n + n} +  O( 2^{n2^n+n})$ \\
 $q=3$ or $q=5$ &  $2n+2$ & $6 (q-1) n q^{nq^n + n -1} +  O( q^{nq^n+n})$\\ 
 $q=4$ or $q \geq 6$ & $2n+2$ & $6 \lceil \log_2(q) \rceil (q-1) n q^{nq^n + n -1} +  O( q^{nq^n+n})$ \\ \hline
\end{tabular}
\caption{Sequentially $n$-complete transformations.} 
\label{ta:sequentially complete}
\end{table} 
\end{theorem}

\begin{proof}
Assume first that $q=4$ or $q \geq 6$, and let $\rho := \lceil \log_2 q \rceil$. Consider the $n$-complete transformation $f$ of size $n+2$ constructed in the proof of Theorem \ref{th:n+2}. Now, define the coordinate functions of $\hat{f}$ by
\begin{align*}
	(x)\hat{f}_i &= \begin{cases}
					 x_{i'} & \text{if } x_b - x_a = \rho + 1\\
					(x_{[n+2]}) f_i & \text{otherwise},
					 \end{cases} &&  (1 \leq i \leq n), \\ \\
	(x)\hat{f}_a &= x_b, \\ \\
	(x)\hat{f}_b&= x_a + 1, \\ \\
	(x)\hat{f}_{[n]'} &= x_{[n]}.
\end{align*}
Intuitively, registers in $[n]'$ maintain a copy of the original configuration of registers in $[n]$; again, registers $a = n+1$ and $b = n+2$ indicate which coordinate function to use but now the position $\rho + 1$ indicates that $f_i$, $1 \leq i \leq n$, must copy back the original values of the input from registers in $[n]'$. 

Let $F^{(i)}$ and $\hat{F}^{(i)}$ be the instructions induced by the coordinate functions $f_i$ and $\hat{f}_i$, respectively. Suppose that we want to sequentially simulate $g^{(1)}, \dots , g^{(\ell)} \in \Tran(A^n)$. Since $f$ is $n$-complete, there exist $h^{(1)}, \dots, h^{(\ell)} \in S_f = \langle F^{(1)}, \dots, F^{(n+2)} \rangle$ such that $\pr_{[n]} \circ g^{(i)} = h^{(i)} \circ \pr_{[n]}$. For $1 \leq i \leq \ell$, define $\hat{h}^{(i)} \in S_{\hat{f}}$ by replacing every instruction $F^{(k)}$ in $h^{(i)}$  by $\hat{F}^{(k)}$. Let 
\[ 
	C := \hat{F}^{(1', \dots, n')} \ \text{ and } \ B := (\hat{F}^{(b)})^{\rho + 1} ( \hat{F}^{(1, \dots, n)} ) \hat{F}^{(a)}. 
\]
Then, for every $1 \leq i \leq \ell$, we have
\[ 
	\pr_{[n]} \circ g^{(1)} \circ \dots \circ g^{(i)} =  (\hat{F}^{(a)} C \hat{h}^{(1)} ) ( B \hat{h}^{(2)} ) (B  \hat{h}^{(3)}) \dots (B \hat{h}^{(i)}) \circ \pr_{[n]}. 
\]

By Theorem \ref{th:n+2}, each $\hat{h}^{(i)}$ has procedural complexity of at most $6 \rho (q-1) n q^{n-1} + O(q^{n})$. Hence, sequences of length $\ell = q^{nq^n}$ have maximal sequential time of $6 \rho (q-1) n q^{nq^n + n -1} +  O( q^{nq^n+n})$.

For $q=3$ or $q=5$, let $\rho :=1$, and use the above construction of $\hat{f}$ with $\tilde{f}$, as in Remark \ref{rk:1}, instead of $f$.

The proof for $q=2$ is very similar. The main difference is that, as the first and second coordinate functions must choose among three possibilities ($i_1$, $t_1$, or $x_{1'}$, and $i_2$, $a_2$, or $x_{2'}$, respectively), a switch consisting of two registers does not suffice; however, a switch of three registers $a, b, c$ is enough for our purposes. More formally, we now define the transformation $\hat{f} \in \Tran(A^{2n+3})$ by
\begin{align*}
	(x)\hat{f}_i &= \begin{cases}
					 x_{i'} & \text{if } x_b \ne x_c\\
					(x_{[n+2]}) f_i & \text{otherwise},
					 \end{cases} &&  (1 \leq i \leq n), \\ \\
	(x)\hat{f}_a&= x_b, \\ \\
	(x)\hat{f}_b&= x_b + 1, \\ \\
	(x)\hat{f}_c&= x_b, \\ \\
	(x)\hat{f}_{[n]'} &= x_{[n]}.
\end{align*}
Using a similar notation as above, define
\[ 
	C := \hat{F}^{(1', \dots, n')} \ \text{ and } \ B := \hat{F}^{(c,b)} ( \hat{F}^{(1, \dots, n)} ) \hat{F}^{(b)}. 
\]
For $1 \leq i \leq \ell$, define $\hat{h}^{(i)} \in S_{\hat{f}}$ by replacing every instruction $F^{(k)}$ in $h^{(i)}$ by $\hat{F}^{(k)}$ for $k \le n+1$ and by replacing every instruction $F^{(b)}$ in $h^{(i)}$ by $\hat{F}^{(b, c)}$. Then, for every $1 \leq i \leq \ell$, we have
\[ \pr_{[n]} \circ g^{(1)} \circ \dots \circ g^{(i)} =  (\hat{F}^{(a)} C \hat{h}^{(1)} ) ( B \hat{h}^{(2)} ) (B  \hat{h}^{(3)}) \dots (B \hat{h}^{(i)}) \circ \pr_{[n]}. \]
The time analysis is similar as before.
\end{proof}


\subsection{Sequentially complete transformations with minimal sequential times}

\begin{theorem} \label{th:seq_complete_min_time}
Let $f$ be a sequentially $n$-complete transformation. Then, $q^{nq^n} \leq  \mathbf{st}^{\min}_{f}$. Conversely, there exists a sequentially $n$-complete transformation $f$ such that $\mathbf{st}^{\min}_{f} = (1 + o(1)) q^{nq^n}$.
\end{theorem}

Clearly, one always needs at least $q^{nq^n}$ updates to compute any sequence of length $q^{nq^n}$, so $q^{nq^n} \leq  \mathbf{st}^{\min}_{f}$. In order to prove the upper bound in the theorem, we need several preliminary results about Gray codes. 

As usual, let $\lvert A \rvert = q$ and $n \geq 2$. An $(n,q)$-\emph{Gray code} is an ordering $(a^{(0)},\dots,a^{(q^n-1)})$ of the states in $A^n$ such that two consecutive states differ by only one coordinate: $d_H(a^{(i - 1 \mod q^n)}, a^{(i)}) = 1$ for all $0 \le i \le q^n-1$, where $d_H$ is the Hamming distance. For any Gray code $G = (a^{(0)}, \dots, a^{(q^n-1)})$, we denote the sequence $C(G) = (c^{(0)}, \dots, c^{(q^n-1)}) \in [n]^{q^n}$ where $c^{(i)} \in [n]$ is the coordinate in which $a^{(i-1 \mod q^n)}$ and $a^{(i)}$ differ. We also denote by $\mathrm{S}$ the successor function of $G$, i.e. $\mathrm{S} \in \Tran(A^n)$ and $(a^{(i)}) \mathrm{S} = a^{(i + 1 \mod q^n)}$.

We first give an example of how to achieve a minimum time of $(2 + o(1))\mathcal{T}$. We view $\Tran(A^n)$ as $[Q]^n$ and use an $(n,Q)$-Gray code to list the functions in $\Tran(A^n)$. We thus denote $\Tran(A^n) = \{g^0, \dots, g^{\mathcal{T}-1}\}$ so that two consecutive functions $g^{i-1}$ and $g^i$ only differ by the $c^i$-th local function: $g^i_{c^i} \ne g^{i-1}_{c^i}$. Independently, we enumerate $\{0, \dots, \mathcal{T} - 1\}$ according to an $(nq^n, q)$-Gray code $\hat{G}$ with $C(\hat{G}) = (\hat{c}^{(0)}, \dots, \hat{c}^{(q^n-1)})$ and successor function $\hat{\mathrm{S}}$. The transformation $f$ is defined as follows. We let $m = 2n + Q + 2$ and $[m] = [n] \cup [n]' \cup \Sigma \cup \{a,b\}$, and we identify $x_\Sigma = (x_{k_1}, \dots, x_{k_{nq^n}})$ with its index in the Gray code $\hat{G}$. Then
\begin{align*}
	(x)f_{[n]} &= (x_{[n]'})g^{x_\Sigma},
	\\ \\
	(x)f_{[n]'} &= x_{[n]}, 
	\\ \\
	(x)f_\Sigma &=	\begin{cases}
	(x_\Sigma) \hat{\mathrm{S}} & \text{if } x_a \ne x_b\\
	0 &\text{if } x_a = x_b,
	\end{cases}
	\\ \\
	(x)f_a &= x_b,
	\\ \\
	(x)f_b &= x_a + 1.
\end{align*}
This is illustrated in Figure \ref{fig:complete_universal_min_time}.

The sequence $g^0, \dots, g^{\mathcal{T}-1}$ is then simulated as follows.
\begin{description}
	\item[Step 1.] Initialisation.
	\begin{description}
		\item[Step 1.1.] Copy $x_{[n]}$ into $[n]'$: $F^{(1', \dots, n')}$.
		
		\item[Step 1.2.] Turn the switch off and reset the counter: $F^{(a, k_1, \dots, k_{nq^n})}$.
		
		\item[Step 1.3.] Turn the switch on: $F^{(b)}$.
	\end{description}
	
	\item[Step 2.] For $0 \le i \le \mathcal{T}-1$
	\begin{description}
		\item[Step 2.1.] Compute $g^i$: $F^{(\hat{c}^{(i)})}$.
		
		\item[Step 2.2.] Increment the counter: $F^{(k_{\hat{c}^{(i)}})}$.
	\end{description}	
\end{description}

\begin{figure}
	\centering
	\begin{tikzpicture}[scale=1.5]
	\begin{scope}
	\node[state] 	(1) at (0,0) {$1$};
	\node[state] 	(2) at (1,0) {$2$};
	\node 			(3) at (2,0) {$...$};
	\node[state] 	(n) at (3,0) {$n$};
	
	\draw (-0.5,1) -- +(4,0) node[below left, font=\small] {$[n]$: First $n$ registers} -- +(4,-1.5) -- +(0,-1.5) -- +(0,0);
	\end{scope}

	\begin{scope}[yshift = -2cm]	
	\node[state] 	(11) at (0,0) {$1'$};
	\node[state] 	(22) at (1,0) {$2'$};
	\node 			(33) at (2,0) {$...$};
	\node[state] 	(nn) at (3,0) {$n'$};
	
	\draw (-0.5,-1) -- +(4,0) node[above left, font=\small] {$[n]'$: Copy of $[n]$} -- +(4,1.5) -- +(0,1.5) -- +(0,0);
	\end{scope}
	
	\path[-latex] 	(1) edge (11)
	(2) edge (22)
	(n) edge (nn);

	\node[draw] (C) at (5,0) {$\Sigma$: counter};
	
	\begin{scope}[xshift = 5cm, yshift = -1.5cm]
	\node[state] 	(a) at (0,0) {$a$};
	\node[state] 	(b) at (0,-1) {$b$};
	
	\path 	(b) edge (a);
	
	\draw (-0.6,1) -- +(1.2,0) node[below left, font=\small] {$S$: Switch} -- +(1.2,-2.5) -- +(0,-2.5) -- +(0,0);
	\end{scope}

	\draw[very thick, -latex] (C) -- (3.5,0);
	\draw[very thick, -latex] (5,-0.5) -- (C);
	\draw[very thick, -latex] (2, -1.5) -- (2, -0.5);
	\end{tikzpicture}
	\caption{The sequentially $n$-complete transformation with minimal sequential time $(2 + o(1))\mathcal{T}$ } \label{fig:complete_universal_min_time}
\end{figure}
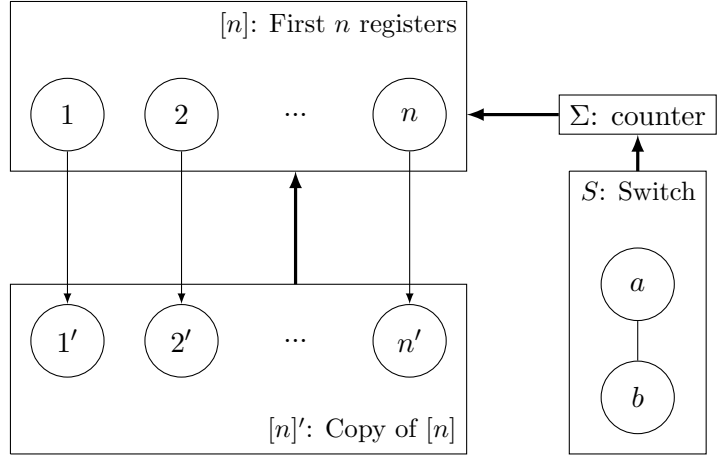

We now turn to proving an upper bound of $(1 + o(1))\mathcal{T}$. In the previous example, we had to update the Gray code counter at each time step, thus $\mathcal{T}$ times in total; the main improvement is to (almost) construct a Gray code where we only need to update the counter $o(\mathcal{T})$ times. 

A \emph{run} of length $l$ for $G$ is a sequence $c^{(i)}, \dots, c^{(i+l-1)}$ of consecutive distinct elements of $C(G)$. We say that $G$ has $r(G)$ runs if $C(G)$ can be partitioned into $r(G)$ runs. For instance, the canonical $(n,2)$-Gray code has $2^{n-1}$ runs. For $n=2$, we have
\begin{center}
\setlength{\tabcolsep}{5pt}
\begin{tabular}{llll}
	$a^{(0)} = 00$, & $a^{(1)} = 01$, & $a^{(2)} = 11$, & $a^{(3)} = 10$,\\
	$c^{(0)} = 1$, & $c^{(1)} = 2$, & $c^{(2)} = 1$, & $c^{(3)} = 2$.
\end{tabular}
\end{center}
For $n=3$, we have
\begin{center}
\setlength{\tabcolsep}{5pt}
\begin{tabular}{llllllll}
	$a^{(0)} = 000$, &  $a^{(1)} = 001$, &  $a^{(2)} = 011$, &  $a^{(3)} = 010$, &  $a^{(4)} = 110$, &  $a^{(5)} = 111$, &  $a^{(6)} = 101$, &  $a^{(7)} = 100$,\\
	$c^{(0)} = 1$, &  $c^{(1)} = 3$, &  $c^{(2)} = 2$, &  $c^{(3)} = 3$, &  $c^{(4)} = 1$, &  $c^{(5)} = 3$, &  $c^{(6)} = 2$, &  $c^{(7)} = 3$.
\end{tabular}
\end{center}
Clearly, any Gray code has at least $q^n/n$ runs; we shall construct $(n,q)$-Gray codes with $o(q^n)$ runs for even $q$. 

\begin{lemma} \label{lem:binary_Gray1}
For any $n$ a power of $2$, there exists an $(n,2)$-Gray code with $o(2^n)$ runs.
\end{lemma}

\begin{proof}
We shall prove the result by induction on $n$. The code $G_2$ is the canonical Gray code. Suppose $G_n = (a^{(0)},\dots,a^{(2^n-1)})$ (or simply written, $0$ up to $2^n-1$), then $G_{2n}$ is given by
\begin{align*}
	& \Big( (0, 0), (0, 1), (1, 1), (1, 2), \dots, (2^n-1, 2^n-1), (2^n-1, 0),\\
	& (2^n-2, 0), (2^n-2, 1), \dots, (2^n - 3, 2^n - 1), (2^n - 3, 0),\\
	& \vdots\\
	& (2, 0), (2, 1), \dots, (1, 2^n-1), (1, 0) \Big).
\end{align*}
There are $2^{n-1}$ rows, each containing $2^{n+1}$ elements. The code $G_4$ is then
\begin{align*}
	G_4 &= \Big( 0000, 0001, 0101, 0111, 1111, 1110, 1010, 1000,\\
	& \quad 1100, 1101, 1001, 1011, 0011, 0010, 0110, 0100 \Big),\\
	C(G_4) &= (2, 4, 2, 3, 1, 4, 2, 3, 2, 4, 2, 3, 1, 4, 2, 3),
\end{align*}
which can be partitioned into six runs (instead of eight for the canonical code).

Let $\Psi(n,d)$ denote the set of indices $i$ such that the next occurrence of $c^{(i)}$ appears at least $d$ indices later. More formally, let $\Gamma_n = (V,E)$ be the directed graph on $V = \{0, \dots, 2^n-1\}$ with arcs $E = \{(i,i+1 \mod 2^n-1) : i \in V\}$ and let $d(i,j)$ be the length of the path from $i$ to $j$ in $\Gamma_n$, then 
$$
	\Psi(n,d) = \{i : 0 < d(i,j) < d \Rightarrow c^{(j)} \ne c^{(i)}\}.
$$
For any $d$, the Gray code $G_n$ has at most 
$$
	|\Psi(n,d)|/d + 2(2^n - |\Psi(n,d)|) + 1 
$$
runs. Indeed, split $\Psi(n,d)$ into sequences $s_1, \dots, s_m$ of consecutive indices, where $m \le 2^n - |\Psi(n,d)| + 1$. Each sequence $s_t$ of length $l_t$ ($1 \le t \le m$) can be partitioned into $\lceil l_t / d \rceil \le l_t/d + 1$ runs, thus requiring at most $|\Psi(n,d)|/d + 2^n - |\Psi(n,d)| + 1$ runs to partition $\Psi(n,d)$. Moreover, the indices outside of $\Psi(n,d)$ can be partitioned into singleton runs; altogether, this yields $|\Psi(n,d)|/d + 2(2^n - |\Psi(n,d)|) + 1$ runs.

Our strategy is then to find a distance $d$ such that $d = \omega(1)$ and $2^n - |\Psi(n,d)| =  o(2^n)$. We have $|\Psi(n,2)| = 2^n$ for all $n$ and by construction,
$$
	|\Psi(2n, 2d)| \ge 2^n |\Psi(n,d)| - 2^n d, 
$$
since the only ones that do not follow the simply doubling pattern are the ones at the end of every row. Denoting the largest power of two less than or equal to $\log_2 n$ as $l$, we then obtain
\begin{align*}
	|\Psi(n, l)| &\ge 2^{n/2} |\Psi(n/2, l/2)| - 2^{n/2} l\\
	&\ge 2^{3n/4} |\Psi(n/4, l/4)| - l( 2^{3n/4 - 1} + 2^{n/2})\\
	&\vdots\\
	&\ge 2^{n - 2n/l} |\Psi(2n/l, 2)| - l(2^{n - 2n/l - l} + \dots + 2^{n/2})\\
	&= 2^n - o(2^n).
\end{align*}
\end{proof} 

\begin{lemma} \label{lem:binary_Gray2}
For any $n$, there exists an $(n,2)$-Gray code with $o(2^n)$ runs.
\end{lemma}

\begin{proof}
This is obtained by the usual ``product'' construction of Gray codes. Let $m$ be the largest power of two less than or equal to $n$. Denote the $(m,2)$-Gray code from Lemma \ref{lem:binary_Gray1} by $( 0, 1, \dots, 2^m -1)$ and an $(n-m,2)$-Gray code by $(0, 1, \dots, 2^{n-m} - 1)$. Now, construct an $(n,2)$-Gray code as follows:
\begin{align*}
	& \Big( (0, 0), \dots, (0, 2^m-1),\\
	& (1, 0), \dots, (1, 2^m-1),\\
	& \vdots\\
	& (2^{n-m}-1, 0), \dots, (2^{n-m} - 1, 2^m - 1) \Big).
\end{align*}
This has at most $2^{n-m} \cdot o(2^m) = o(2^n)$ runs.
\end{proof}

\begin{lemma} \label{lem:even_Gray}
For any even $q$ and any $n$, there exists an $(n,q)$-Gray code with $o(q^n)$ runs.
\end{lemma}

\begin{proof}
Here again, we use a ``product'' construction, viewing each element of $[q]^n (=A^n)$ as an element of $[p]^n \times [2]^n$, where $p =  \frac{q}{2}$. We then combine any $(n,p)$-Gray code with the $(n,2)$-Gray code from Lemma \ref{lem:binary_Gray2} as follows:
\begin{align*}
	& \Big( (0, 0), \dots, (0, 2^n-1),\\
	& (1, 0), \dots, (1, 2^n-1),\\
	& \vdots\\
	& (p^n - 1, 0), \dots, (p^n - 1, 2^n - 1) \Big).
\end{align*}
Clearly, this has at most $p^n \cdot o(2^n) = o(q^n)$ runs.
\end{proof}

For odd $q$, we do not use a Gray code. Instead, an $(n,q)$-\emph{pseudo-Gray code} of length $L$ is a sequence $P = (p^{(0)}, \dots, p^{(L-1)})$ of elements of $[q]^n$ such that every element of $[q]^n$ appears in the sequence and any two consecutive elements of the sequence only differ by one coordinate. (A pseudo-Gray code is a Gray code if every element appears exactly once.) Runs are defined for pseudo-Gray codes in the natural way and the number of runs is still denoted $r(P)$; the \emph{redundancy} $R(P)$ of a pseudo-Gray code is $R(P) = r(P) + L - q^n$.

\begin{lemma} \label{lem:odd_Gray}
For any $q$ and any $n$, there exists an $(n,q^{q^n})$-pseudo-Gray code with redundancy $o(q^{nq^n})$.
\end{lemma}

\begin{proof}
If $q$ is even, we use the Gray code from Lemma \ref{lem:even_Gray}. Suppose that $q$ is odd. Then $Q:=q^{q^n}$ is odd, so, again by Lemma \ref{lem:even_Gray}, there is an $(n,Q-1)$-Gray code $G$ with $o(Q^{n})$ runs. We shall construct an $(n,Q)$-pseudo-Gray code by using $G$ first, and then enumerating all the remaining states in $[Q]^n$. It takes at most $n$ steps to go from of these remaining states to another, and there are $Q^{n} - (Q-1)^n$ of them. Thus, the redundancy of this pseudo-Gray code is at most
\begin{align*}
	2n (Q^n - (Q-1)^n) + o(Q^n) &= 2n (q^{nq^n} - (q^{q^n} - 1)^n) + o(Q^n) \\
	& \le 2n  \cdot n q^{(n-1)q^n} + o(Q^n) = o(Q^n).
\end{align*}
\end{proof}

Finally, we may prove Theorem \ref{th:seq_complete_min_time}.


\begin{proof}[Proof of Theorem \ref{th:seq_complete_min_time}]
We explicitly construct the sequentially $n$-complete transformation $f$ of the statement of the theorem. Let $Q = q^{q^n}$ and $P = (p^{(0)} = \id, \dots, p^{(\mathcal{T} - 1)})$ be the $(n, Q)$-pseudo Gray code of Lemma \ref{lem:odd_Gray}. We use the notation $C(P) = (c^{(0)}, \dots, c^{(\mathcal{T} - 1)})$, which is partitioned into $r = r(P)$ runs $R_1 = (c^{(0)}, \dots, c^{(\rho_1 - 1)}), \dots, R_r = (c^{(\rho_{r-1})}, \dots, c^{(\mathcal{T} - 1)})$ and
\begin{align*}
	\tau &: [r] \times [n] \to \{0,\dots, \mathcal{T}-1\}\\
	(s, i) \tau &= \begin{cases}
		\lambda & \text{if } \exists \lambda : i = c_\lambda \in R_s\\
		0 &\text{if } i \notin R_s.
	\end{cases}
\end{align*}

Let $\bar{G} = (\bar{a}^{(0)}, \dots, \bar{a}^{(q^\sigma -1)})$ be a $(\sigma, q)$-Gray code, where $\sigma = \lceil \log_q r \rceil + 1$. Denoting $\Sigma = (k_1, \dots, k_\sigma)$, again we identify $x_\Sigma$ with its index in $\bar{G}$. We also use $C(\bar{G}) = (\bar{c}_0, \dots, \bar{c}_{q^\sigma - 1})$ and the successor function for this code is $\bar{\mathrm{S}}$.

Let $m = 2n + \sigma + 2$ and $[m] = [n] \cup [n]' \cup \Sigma \cup \{a,b\}$. Then the transformation $f$ is defined as follows:
\begin{align*}
	(x) f_i &= (x_{[n]'}) p_i^{( ( x_\Sigma , \; i)  \tau )}, && (1 \le i \le n),\\ \\
	(x) f_{[n]'} &= x_{[n]}, \\ \\
	(x) f_\Sigma &=	\begin{cases}
					(x_\Sigma) \bar{\mathrm{S}} &\text{if } x_a = x_b\\
					1 & \text{if } x_a \ne x_b,
				\end{cases}
				&& (1 \le i \le \sigma),\\ \\
	(x) f_a &= x_b,\\ \\
	(x) f_b &= x_b + 1.
\end{align*}
Intuitively, registers in $[n]'$ maintain a copy of the original configuration of registers in $[n]$; registers in $\Sigma$ form a counter indicating the run number in the pseudo-Gray code $P$; registers $a$ and $b$ form a reset switch for the run counter. This transformation is illustrated in Figure \ref{fig:complete_universal_min_time_optimal}.

The program computing $P = (p^{(0)} = \id, \dots, p^{(\mathcal{T} - 1)})$ in order goes as follows.
\begin{description}
	\item[Step 1.] Make a copy of the first $n$ registers: $F^{(1', \dots, n')}$.
	
	\item[Step 2.] Reset switch on: $F^{(a)}$.
	
	\item[Step 3.] Reset run counter: $F^{(k_1, \dots, k_\sigma)}$.
	
	\item[Step 4.] Reset switch off: $F^{(b)}$.
	
	\item[Step 5.] For $s$ from $1$ to $r$ do:
	\begin{description}
		\item[Step 5.1.] Compute $p^{\left(\rho_{(s-1)}\right)}$ to $p^{(\rho_s-1)}$:  $F^{ \left( c^{ ( \rho_{(s-1)} ) }, \dots, c^{ ( \rho_s - 1 ) } \right) }$;
		
		\item[Step 5.2.] Increment the run counter: $F^{(k_{\bar{c}_{s+1}})}$.
	\end{description}
\end{description}
Total time:
$$
	n + 1 + \sigma + 1 + \sum_{s = 1}^r (\rho_s - \rho_{s-1} + 1) = \mathcal{T} + r + O(\log r) = \mathcal{T} + o(\mathcal{T}).
$$
We finally prove that this transformation is sequentially complete. Let $p^{(i_1)}, \dots, p^{(i_\ell)} \in \Tran(A^n)$ be a sequence of transformations. It is clear that applying Step 1 and then repeating $\ell$ times Steps 2 to 5 will simulate $\ell$ times the full sequence $p^{(0)}, \dots, p^{(\mathcal{T} - 1)}$. As such, $p^{(i_1)}$ is simulated during the first iteration, $p^{(i_2)}$ during the next, and so on until $p^{(i_\ell)}$.
\end{proof}

\begin{figure}
	\centering
	\begin{tikzpicture}[scale=1.5]
	\begin{scope}
		\node[state] 	(1) at (0,0) {$1$};
		\node[state] 	(2) at (1,0) {$2$};
		\node 			(3) at (2,0) {$...$};
		\node[state] 	(n) at (3,0) {$n$};
		
		\draw (-0.5,1) -- +(4,0) node[below left, font=\small] {$[n]$: First $n$ registers} -- +(4,-1.5) -- +(0,-1.5) -- +(0,0);
	\end{scope}

	\begin{scope}[yshift = -2cm]	
		\node[state] 	(11) at (0,0) {$1'$};
		\node[state] 	(22) at (1,0) {$2'$};
		\node 			(33) at (2,0) {$...$};
		\node[state] 	(nn) at (3,0) {$n'$};
		
		\draw (-0.5,-1) -- +(4,0) node[above left, font=\small] {$[n]'$: Copy of $[n]$} -- +(4,1.5) -- +(0,1.5) -- +(0,0);
	\end{scope}
	
	\path[-latex] 	(1) edge (11)
	(2) edge (22)
	(n) edge (nn);

	\node[draw] (C) at (4.5,0) {$\Sigma$: counter};
	
	\begin{scope}[xshift = 4.5cm, yshift = -1.5cm]
		\node[state] 	(a) at (0,0) {$a$};
		\node[state] 	(b) at (0,-1) {$b$};
		
		\path 	(b) edge (a);
		
		\draw (-0.6,1) -- +(1.2,0) node[below left, font=\small] {$S$: Switch} -- +(1.2,-2.5) -- +(0,-2.5) -- +(0,0);
	\end{scope}
	
	\begin{scope}[xshift = 6cm, yshift = 0cm]
		\node[state] 	(d1) at (0,0) {$d_1$};
		\node[state] 	(d2) at (1,0) {$d_2$};
		\node 			(d3) at (2,0) {$...$};
		\node[state] 	(dk) at (3,0) {$d_{\hat{k}}$};
		
		\draw (-0.5,1) -- +(4,0) node[below left, font=\small] {$D$: Data bits} -- +(4,-1.5) -- +(0,-1.5) -- +(0,0);
	\end{scope}
	\begin{scope}[xshift = 6cm, yshift = -1.5cm]
		\node[state] 	(p1) at (0,0) {$p_1$};
		\node[state] 	(p2) at (1,0) {$p_2$};
		\node 			(p3) at (2,0) {$...$};
		\node[state] 	(pr) at (3,0) {$p_{\hat{r}}$};
		
		\draw (-0.5,1) -- +(4,0) node[below left, font=\small] {$P$: Parity bits} -- +(4,-1.5) -- +(0,-1.5) -- +(0,0);
	\end{scope}

	\draw[very thick, -latex] (C) -- (3.5,0);
	\draw[very thick, -latex] (4.5,-0.5) -- (C);
	\draw[very thick, -latex] (2, -1.5) -- (2, -0.5);
	\draw[very thick, -latex] (5.5, 0.5) -- (3.5, 0.5);
	\end{tikzpicture}
	\caption{The sequentially $n$-complete transformation with time $(2 + o(1))\mathcal{T}$ } \label{fig:complete_universal_min_time_optimal}
\end{figure}
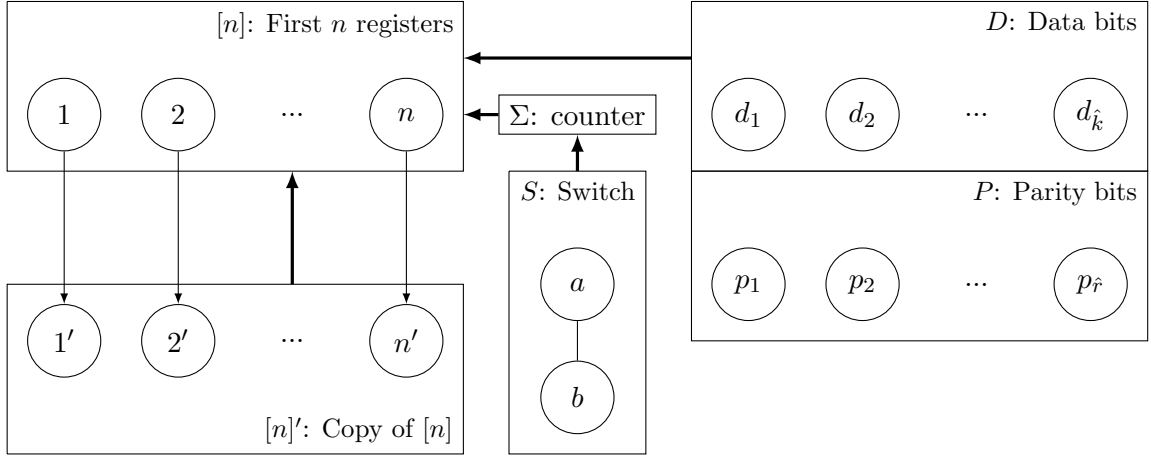

\begin{theorem} \label{th:complete_max_time}
Let $f$ be a sequentially $n$-complete transformation. Then, $nq^{nq^n} \leq  \mathbf{st}^{\max}_{f}$. Conversely, there exists a sequentially $n$-complete transformation $f$ such that 
$\mathbf{st}^{\max}_f = (n + 1 + o(1))q^{nq^n}$.
\end{theorem}

\begin{proof}
Viewing any coordinate function $A^n \to A$ as an element in $\mathbb{Z}_Q$, with $Q = q^{q^n} \ge 4$, we give an ordering of $\Tran(A^n) \cong \mathbb{Z}_Q^n$ such that any two consecutive transformations differ in all $n$ coordinate functions:
\begin{align*}
&	\Big(	(0, \dots, 0), (1, \dots, 1), \dots, (Q-1, \dots, Q-1),\\ 
&			(1,0, \dots, 0), (2,1, \dots, 1), \dots, (0, Q-1, \dots, Q-1),\\
&			\vdots\\
&			(0,1, \dots, 0), (1,2, \dots, 1), \dots, (Q-1, 0, \dots, Q-1),\\
&			\vdots\\
&			(Q-1, \dots, Q-1, 0), (0, \dots, 0,1), \dots, (Q-2, \dots, Q-2, Q-1) \Big).
\end{align*}
Clearly, the time of simulation of such a sequence of transformations is at least $nq^{nq^n}$, so $nq^{nq^n} \leq \mathbf{st}^{\max}_{f}$.

The construction of $f$ is based on Hamming codes, and describes a whole sequence of transformation of $A^n$ at once. Firstly, for $\hat{k}: = Q^{nQ}$, $\hat{r} = \lceil \log_2 \hat{k} \rceil + 1$ and $\hat{n} := \hat{k} + \hat{r}$, consider the $(\hat{n}, \hat{k}, 3)$-shortened Hamming code $\mathcal{C}$ in systematic form. Let $M \in \GF(2)^{\hat{k} \times \hat{n}}$ be its generator matrix and, for $H \subseteq [\hat{n}]$, let $M_H$ be the matrix formed with the $H$-th columns of $M$. This is a binary code, so in order to use it, let
\begin{align*}
	\mathrm{odd} &: A \to \GF(2)\\
	(a) \mathrm{odd} &= \begin{cases}
		1 & \text{if } a \equiv 1 \mod 2\\
		0 & \text{if } a \equiv 0 \mod 2,
	\end{cases}\\ \\
	\mathrm{err} &: A \to A\\
	(a) \mathrm{err} &= \begin{cases}
		a + 1 & \text{if } a < q-1\\
		a - 1 & \text{if } a = q-1.
	\end{cases}
\end{align*}
By applying them component-wise, we extend these functions to $\mathrm{odd} : A^k \to \GF(2)^k$ and $\mathrm{err} : A^k \to A^k$, for any $k \ge 1$. The shortened Hamming code can correct one error, thus let 
\begin{align*}
	\mathrm{dec} &: \GF(2)^{\hat{n}} \to \{0,\dots, \hat{n}\}\\
	(v) \mathrm{dec} &= \begin{cases}
		j & \text{if } v = c + e^j \text{ for some } c \in \mathcal{C}\\
		0 & \text{otherwise}.
	\end{cases}
\end{align*}

We denote $(\Tran(A^n))^Q = \{\Gamma^{(1)}, \dots, \Gamma^{(\hat{k})}\}$, where $\Gamma^{(j)} = (g^{(j;0)}, \dots, g^{(j;Q-1)})$, and $r = \lceil (\lceil \log_2 \hat{k} \rceil + 1) / \rho \rceil$. Let $m := 2n+ \hat{k} + \hat{r} + q^n + 2$ and $[m] = [n] \cup [n]' \cup D \cup P \cup \Sigma \cup \{a,b\}$, where $D = (d_1, \dots, d_{\hat{k}})$ are the indices corresponding to the data bits and $P = (p_1, \dots, p_{\hat{r}})$ correspond to the parity bits instead, and where $\Sigma = (k_1, \dots, k_{q^n})$. Again, we identify $x_\Sigma$ with its index in the $(q^n,q)$-Gray code $G = (a^{(0)}, \dots, a^{(Q-1)})$, and let $C(G) = (c^{(0)}, \dots c^{(Q -1)})$. The index $x_\Sigma$ shall work as a counter to decide which transformation will be simulated. For $x \in A^m$, define $(x)\tau :=(x_{D \cup P}) \mathrm{odd} \, \mathrm{dec}$. The sequentially $n$-complete transformation $f$ is given as follows:
\begin{align*}
	(x) f_i &= \begin{cases}
	 x_i, & \text{if } (x)\tau > \hat{k}, \\
	 (x_{[n]'}) g_i^{\left( (x)\tau ; \ x_\Sigma \right)}, & \text{otherwise},
	 	 \end{cases}  && (1 \le i \le n), \\ \\
	(x) f_{[n]'} &= x_{[n]}, \\ \\
	(x) f_{d_i} &= (x_{d_i}) \mathrm{err}, && (1 \le i \le \hat{k}), \\ \\
	(x) f_P &= \left( (x_D) \; \mathrm{odd}  \right) M_P, \\ \\
	(x) f_\Sigma &=	\begin{cases}
				(x_\Sigma) \mathrm{S} & \text{if } x_a = x_b,  \\
				0 &\text{if } x_a \ne x_b,
				\end{cases}  
				 \\ \\
	(x) f_a &= x_b,\\ \\
	(x) f_b &= x_b + 1.
\end{align*}

Let $\lambda \ge 1$ and $\Lambda = Q\lambda$. The sequence $g^{(i_1;0)}, \dots, g^{(i_1;Q-1)}, \dots,g^{(i_\lambda;0)}, \dots, g^{(i_\lambda; Q-1)}$ of length $\Lambda$ is simulated as follows:
\begin{description}
	\item[Step 1.] Make a copy of the first $n$ registers: $F^{(1', \dots, n')}$.
	
	\item[Step 2.] Encode $x_D$ into a codeword of $\mathcal{C}$: $F^{(p_1, \dots, p_{\hat{r}})}$.
	
	\item[Step 3.] Reset the counter $x_\Sigma$: $F^{(a)} \circ F^{(b)} \circ F^{(k_1, \dots, k_{q^n})} \circ  F^{(a)}$.
	
	\item[Step 4.] For $\mu$ from $1$ to $\lambda$ do:
	\begin{description}
		\item[Step 4.1.] Add an error to $\Gamma^{(i_\mu)} = (g^{(i_\mu;0)}, \dots, g^{(i_\mu; Q-1)}) $: $F^{(d_\mu)}$.
		
		\item[Step 4.2.] For $\sigma$ from $0$ to $Q-1$ do:
		\begin{description}
			\item[Step 4.2.1.] Compute $g^{(i_\mu;\sigma)}$: $F^{(1, \dots, n)}$.
			
			\item[Step 4.2.2.] Increment the counter $\Sigma$ according to the Gray code: $F^{( k_{c^{(\sigma+1)}} )}$. 
		\end{description}
		
		\item[Step 4.3.] Remove the error: $F^{(d_\mu)}$.
	\end{description}
\end{description}
For $\Lambda = \mathcal{T}$, the time to simulate this sequence is then given by 
$$
	n + \hat{r} + (q^n + 3) + Q^{n-1}( Q(n+1) + 2) = (n + 1 + o(1))\mathcal{T}.
$$
\end{proof}

Note that the sequentially $n$-complete transformation $\hat{f}$ of size $2n+2$ (or $2n+3$ when $q=2$) given in Theorem \ref{th:space_2n+2} does not have a very high sequential time compared with sequentially $n$-complete transformations of minimal sequential time; indeed, we may see that $\mathbf{st}^{\max}_{\hat{f}}$ is equal $O({\bf st}^{\max}_f \log {\bf st}^{\max}_f)$, with $f$ as in Theorem \ref{th:complete_max_time}.



\section{Simulation of transformations in parallel} \label{sec:asy}

So far, we have looked at sequential updates (i.e. one register at a time). This is a strong constraint for MC: if we were allowed to update all registers at once, then any function could be computed in only one time step. However, in our model of complete simulation, this is actually a strength and a necessity.

We extend our framework to consider the following type of simulations.

\begin{definition}[Parallel simulation]
Let $m \geq n \geq 1$. We say that $f: A^m \to A^m$ \emph{simulates in parallel} $g : A^n \to A^n$ if there exists $h \in \langle f \rangle$ such that $\pr_{[n]} \circ  g = h \circ \pr_{[n]}$.
\end{definition}

We also consider the slightest form of asynchronism in sequential simulations. For $f \in A^m \to A^m$, define $F^{(-m)} : A^m \to A^m$ by
\[ (x)F^{(-m)} := \left( (x)f_1,  (x)f_2, \dots,  (x)f_{m-1},  x_m \right). \]

\begin{definition}[Quasi-parallel sequential simulation]
Let $m \geq n \geq 1$. We say that $f: A^m \to A^m$ \emph{sequentially simulates in quasi-parallel} $g^{(1)}, \dots, g^{(\ell)} \in \Tran(A^n)$ if there exist $h^{(1)}, \dots, h^{(\ell)} \in \langle F^{(-m)}, F^{(m)} \rangle$ such that, for all $1 \leq i \leq \ell$,
\[  \pr_{[n]} \circ g^{(i)}  =  h^{(1)} \circ \dots \circ h^{(i)} \circ \pr_{[n]}, \] 
and the instruction $F^{(m)}$ appears at most once in the program $h^{(1)} \circ \dots \circ h^{(\ell)}$.
\end{definition}

Say that $f: A^m \to A^m$  is a \emph{quasi-parallel $n$-complete transformation} if it may sequentially simulate in quasi-parallel any finite sequence in $\Tran(A^n)$.

\begin{theorem} \label{th:quasi-parallel}
For any $m \geq n \ge 1$, there is no transformation $f : A^m \to A^m$ that may simulate in parallel every transformation in $\Tran(A^n)$. However, for any $n \ge 1$, there exists a quasi-parallel $n$-complete transformation. 
\end{theorem}

\begin{proof}
Suppose that $f : A^m \to A^m$ may simulate in parallel every transformation in $\Tran(A^n)$. For $a,b \in A^n$, $a \neq b$, consider the constant transformations $g^{(a)}, g^{(b)} \in \Tran(A^n)$ defined by $(x)g^{(a)} = a$ and $(x)g^{(b)} = b$, for all $x \in A^n$. Then, by definition of parallel simulation, there exist integers $k_a < k_b$ such that, for all $x \in A^m$, 
\[ (x) f^{k_a}\circ \pr_{[n]} = (x)\pr_{[n]}\circ g^{(a)} = a \ \text{ and } \ (x) f^{k_b} \circ \pr_{[n]} = (x)\pr_{[n]} \circ g^{(b)} = b. \]
But now we obtain a contradiction:
\[  b = (x) f^{k_b} \circ \pr_{[n]} = \left( (x) f^{k_b - k_a} \right) f^{k_a} \circ \pr_{[n]}  = a. \]

Denote $\Tran(A^n) = \{p^2 , \dots, p^{\mathcal{T}+1} = \mathrm{id}\}$; moreover, let $p^c = \id$ for all $c \notin \{2, \dots, \mathcal{T}+1\}$. Let $m = (\mathcal{T}+2)n + \lceil \log_q (\mathcal{T}+1) \rceil + 2$ and let us decompose $[m]$ as $[n] \cup [n]_0 \cup \dots \cup [n]_\mathcal{T} \cup C \cup \{a, m\}$; moreover we let $[n]_{-1} = [n]$ and we denote the elements of $[n]_j$ as $(j,1), \dots, (j,n)$ for all $j$. Now we shall construct a quasi-parallel $n$-complete transformation $f$. The registers $[n]_0$ to $[n]_\mathcal{T}$ will hold successive copies of $x_{[n]}$; the registers in $C$ are a counter for $\Tran(A^n)$, as such we identify $x_C$ with its lexicographic index $c$. Define $f$ by:
\begin{alignat*}{2}
	 (x)f_{[n]} &= (x_{[n]_{c-1}}) p^{c+1} ,
	 \\ \\
	 (x)f_{[n]_j} &= x_{[n]_{j-1}}, &\qquad& 0 \le j \le \mathcal{T},
	 \\ \\
	 (x)f_C &= 
		\begin{cases}
		c + 1 \mod \mathcal{T}+2 &\text{if } x_m = x_a\\
		2 & \text{if } x_m \ne x_a,
		\end{cases}
		\\ \\
	 (x)f_a &= x_m,
	 \\ \\
	 (x)f_m &= x_a + 1.
\end{alignat*}

We prove that $f$ is a quasi-parallel $n$-complete transformation by giving a program simulating the sequence $p^1, \dots, p^{\mathcal{T}+1}$ repeated $l$ times for any $l \ge 1$.
\begin{description}
	\item[Step 1.] Do $F^{(-m)}$. This copies $x_{[n]}$ into $x_{[n]_0}$ and turns the switch off: $x_a = x_m$.
	
	\item[Step 2.] Do $F^{(m)}$. This turns the switch on: $x_m \ne x_a$.
	
	\item[Step 3.] Do $F^{(-m)}$. This resets the counter to $c = 2$, and turns the switch off: $x_a = x_m$. Note that the original contents of the first $n$ registers are now contained in $x_{[n]_1}$. 
	
	\item[Step 4.]  Do $\left( F^{(-m)} \right)^\mathcal{T}$. At each iteration, this increases the counter $c = (c+1) \mod \mathcal{T} + 2$, and so it computes the whole sequence $p^2, \dots, p^{\mathcal{T}+1}$. Note that the original value of $x_{[n]}$ is back in $x_{[n]}$.
	
	\item[Step 5.] Do $\left( F^{(-m)} \right)^{(l-1)(\mathcal{T} + 2)}$ in order to compute $l-1$ iterations of $p^1, \dots, p^{\mathcal{T}+1}$.
\end{description}

Now, as in the proof of Theorem \ref{th:seq_complete_min_time}, in order to simulate an arbitrary sequence $p^{i_1}, \dots, p^{i_l} \in \Tran(A^n)$, we use the program above to simulate $l$ times the full sequence $p^2, \dots, p^{\mathcal{T}+1}$. 
\end{proof}

\begin{figure}
	\centering
	\begin{tikzpicture}[scale=1.5]
	\begin{scope}
		\node[state] 	(1) at (0,0) {$1$};
		\node[state] 	(2) at (1,0) {$2$};
		\node 			(3) at (2,0) {$...$};
		\node[state] 	(n) at (3,0) {$n$};
		
		\draw (-0.5,1) -- +(4,0) node[below left, font=\small] {$[n]$: First $n$ registers} -- +(4,-1.5) -- +(0,-1.5) -- +(0,0);
	\end{scope}

	\begin{scope}[yshift = -1.5cm]	
		\node[state] 	(11) at (0,0) {$10$};
		\node[state] 	(21) at (1,0) {$20$};
		\node 			(31) at (2,0) {$...$};
		\node[state] 	(n1) at (3,0) {$n0$};
		
		\node			(N1) at (2,-0.8) {$[n]_0$};
		
		\draw (-0.5,-1) -- +(4,0) -- +(4,1.5) -- +(0,1.5) -- +(0,0);
	\end{scope}

	\path[-latex] 	(1) edge (11)
(2) edge (21)
(n) edge (n1);

	\begin{scope}[yshift = -3.5cm]	
		\node[state] 	(12) at (0,0) {$11$};
		\node[state] 	(22) at (1,0) {$21$};
		\node 			(32) at (2,0) {$...$};
		\node[state] 	(n2) at (3,0) {$n1$};

		\node			(N2) at (2,-0.8) {$[n]_1$};
		
		\draw (-0.5,-1) -- +(4,0) -- +(4,1.5) -- +(0,1.5) -- +(0,0);
	\end{scope}

	\path[-latex] 	(11) edge (12)
(21) edge (22)
(n1) edge (n2);

	\begin{scope}[yshift = -6cm]	
		\node[state] 	(1F) at (0,0) {$1\mathcal{T}$};
		\node[state] 	(2F) at (1,0) {$2\mathcal{T}$};
		\node 			(3F) at (2,0) {$...$};
		\node[state] 	(nF) at (3,0) {$n\mathcal{T}$};
		
		\draw (-0.5,-1) -- +(4,0) node[above left, font=\small] {$[n]_\mathcal{T}$: Ultimate copy of $[n]$} -- +(4,1.5) -- +(0,1.5) -- +(0,0);
	\end{scope}

	\begin{scope}[yshift = -5cm]
		\node	(1d) at (0,0) {$\vdots$};
		\node	(2d) at (1,0) {$\vdots$};
		\node	(nd) at (3,0) {$\vdots$};
	\end{scope}

	\path[-latex]	(12) edge (1d)
	(22) edge (2d)
	(n2) edge (nd)
	(1d) edge (1F)
	(2d) edge (2F)
	(nd) edge (nF);
	
	\node[draw] (C) at (6,0.5) {$C$: counter};

	\begin{scope}[xshift = 6cm, yshift = -2cm]
		\node[state] 	(a) at (0,0) {$a$};
		\node[state] 	(b) at (0,-1) {$m$};
		
		\path 	(b) edge (a);
		
		\draw (-0.6,1) -- +(1.2,0) node[below left, font=\small] {$S$: Switch} -- +(1.2,-2.5) -- +(0,-2.5) -- +(0,0);
	\end{scope}

	\draw[very thick, -latex] (C) -- (3.5,0.5);
	\draw[very thick, -latex] (6,-1) -- (C);
	\path (3.5, -1.5) edge[very thick] (4,-1.5)
	(4,-1.5) edge[very thick] (4,-0.7) 
	(4,-0.4) edge[very thick, -latex] (3.5, -0.4);
	\path (3.5, -3.5) edge[very thick] (4,-3.5)
	(4,-3.5) edge[very thick] (4,-0.2) 
	(4,-0.2) edge[very thick, -latex] (3.5, -0.2);
	\path (3.5, -6) edge[very thick] (4,-6)
	(4,-6) edge[very thick] (4,0) 
	(4,0) edge[very thick, -latex] (3.5, 0);
	\end{tikzpicture}
	\caption{The quasi-parallel $n$-complete transformation of Theorem \ref{th:quasi-parallel}} \label{fig:quasi-parallel}
\end{figure}
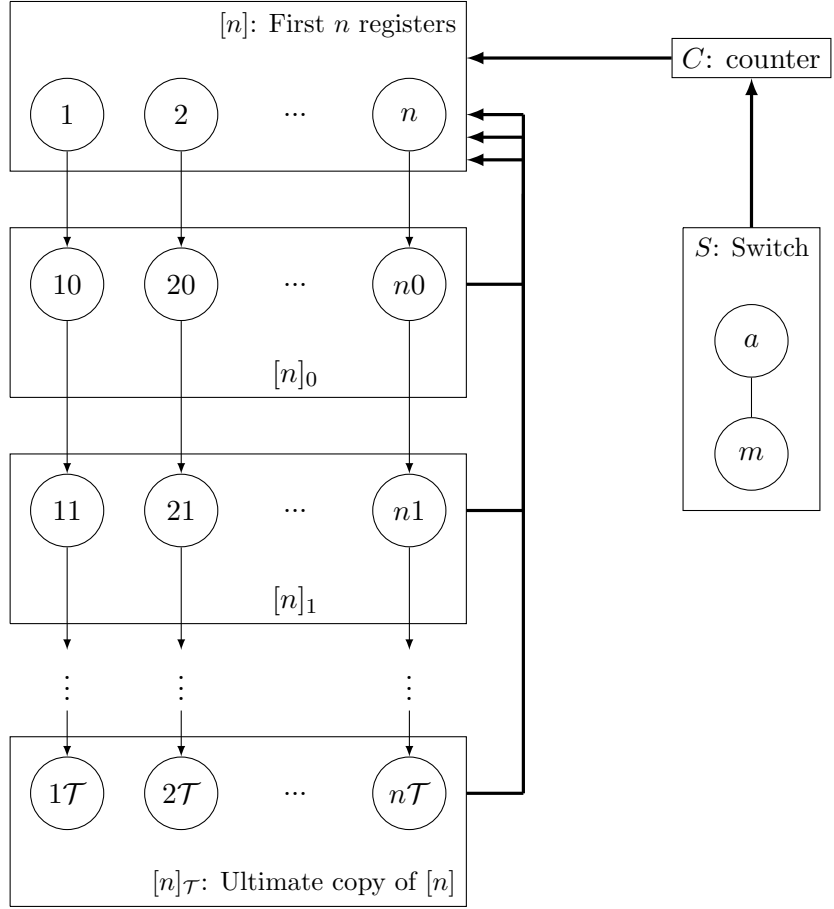


\section{Conclusions and future work}

In this paper, we studied $n$-complete automata networks over an alphabet $A$. These are transformations $f : A^m \to A^m$ which can simulate any transformation of $A^n$ by updating the registers using the coordinate functions of $f$. We showed that there is no $n$-complete transformation of size $n$, but we can always construct one of size $n+1$. We constructed a $n$-complete transformation which simulates any transformation with a maximal time of $2n$, and conjectured that this is the optimal time. We also studied transformations of $A^m$ which can sequentially simulate every sequence of transformations of $A^n$. We proved that the optimal time for such a transformation is between $s := q^{nq^n}$ and $c s$, where $c$ is a constant larger than $1$, and that its optimal size is at least $2n$. 
Finally, we established that there is no complete transformation updating all registers in parallel, but there exists one that updates all but one register in parallel. 

A natural generalization to our notion of $n$-completeness is to allow the transformation $f$  to work with a larger alphabet that the transformations it has to simulate. There are several definitions of simulation in this case. We could fix a projection $\mu$ from alphabet $B$ to $A$ and say that $f \in \Tran(B^n)$ simulates $h \in \Tran(A^n)$ if there is an update order $w$ such that,
\[ F^w \circ \varphi = \varphi \circ h \text{ with } \varphi: x \mapsto ((x_1)\mu, (x_2)\mu, \dots, (x_n)\mu ).\]
With this definition we can already construct an $n,q$-complete transformation $f \in \Tran(B^n)$, for any $q=\vert A \vert $ and $n \geq 3$, with $\vert B \vert = 2q$.
Alternatively, we could say that $f \in \Tran(B^n)$ simulates $h \in \Tran(A^n)$ if there is a certain update order $w$ such that,
\[ \forall x \in A^n,\ (x)F^w = (x)h.\]
With this definition, we can construct a $n,q$-complete transformation $f \in \Tran(B^n)$, for any $q$ and $n \geq 3q$, with $\vert B \vert = q+1$. We could modify these definitions by considering, for example, any projection from $B^n$ to $A^n$.

\section{Acknowledgment}

We would like to thank Adrien Richard and Jean-Paul Comet for stimulating discussions. We also thank both of the anonymous referees of this paper for their valuable comments; in particular, the first referee independently found constructions of $n$-complete transformations of size $n+1$ when $q=2$ and $q \geq 3$. This work was supported by the EPSRC grant EP/K033956/1.

\bibliographystyle{amsplain}

\bibliography{MemorylessBib}

\providecommand{\bysame}{\leavevmode\hbox to3em{\hrulefill}\thinspace}
\providecommand{\MR}{\relax\ifhmode\unskip\space\fi MR }
\providecommand{\MRhref}[2]{%
  \href{http://www.ams.org/mathscinet-getitem?mr=#1}{#2}
}
\providecommand{\href}[2]{#2}
\begin{thebibliography}{10}

\bibitem{ACLY00}
R.~Ahlswede, N.~Cai, S.-Y.R. Li, and R.W. Yeung, \emph{Network information
  flow}, IEEE Trans. Inform. Theory \textbf{45} (2000), no.~4, 1204--1216.

\bibitem{B18}
F.~Bridoux, \emph{Sequentialization and procedural complexity in automata
  networks}, LNCS, Springer-Verlag, 2018, p.~to appear.

\bibitem{Bu96}
S.~Burckel, \emph{Closed iterative calculus}, Theoret. Comput. Sci.
  \textbf{158} (1996), 371--378.

\bibitem{Bu04}
\bysame, \emph{Elementary decompositions of arbitrary maps over finite sets},
  J. Symbolic Comput. \textbf{37} (2004), no.~3, 305--310.

\bibitem{BGT09}
S.~Burckel, E.~Gioan, and E.~Thom\'e, \emph{Mapping computation with no
  memory}, Proc. International Conference on Unconventional Computation (Ponta
  Delgada, Portugal), 2009, pp.~85--97.

\bibitem{BGT14}
\bysame, \emph{Computation with no memory, and rearrangeable multicast
  networks}, Disc. Math. Theor. Comput. Sci. \textbf{16} (2014), 121--142.

\bibitem{BM00}
S.~Burckel and M.~Morillon, \emph{Three generators for minimal writing-space
  computations}, Theor. Inform. Appl. \textbf{34} (2000), 131--138.

\bibitem{BM04a}
\bysame, \emph{Quadratic sequential computations of boolean mappings}, Theory
  Comput. Sys. \textbf{37} (2004), no.~4, 519--525.

\bibitem{BM04b}
\bysame, \emph{Sequential computation of linear boolean mappings}, Theoret.
  Comput. Sci. \textbf{314} (2004), 287--292.

\bibitem{CFG14b}
P.~J. Cameron, B.~Fairbairn, and M.~Gadouleau, \emph{Computing in matrix groups
  without memory}, Chicago J. Theoret. Comput. Sci. (2014), no.~08, 1--16.

\bibitem{CFG14a}
\bysame, \emph{Computing in permutation groups without memory}, Chicago J.
  Theoret. Comput. Sci. (2014), no.~07, 1--20.

\bibitem{DN05}
P.~D\H{o}m\H{o}si and C.~L. Nehaniv, \emph{Algebraic theory of automata
  networks: An introduction}, SIAM Monographs on Discrete Mathematics and
  Applications, 2005.

\bibitem{E91}
Z.~\'Esik, \emph{A note on isomorphic simulation of automata by networks of
  two-state automata}, Discrete Appl. Math. \textbf{30} (1991), 77--82.

\bibitem{GR15}
M.~Gadouleau and S.~Riis, \emph{Memoryless computation: New results,
  constructions, and extensions}, Theoret. Comput. Sci. \textbf{562} (2015),
  129--145.

\bibitem{AN87}
F.~F. Soulie, Y.~Robert, and M.~Tchuente, \emph{Automata networks in computer
  science: theory and application}, Manchester University Press, 1987.

\bibitem{T79}
M.~Tchuente, \emph{Parallel calculation of a linear mapping on a computer
  network}, Linear Algebra Appl. \textbf{28} (1979), 223--247.

\bibitem{T82}
\bysame, \emph{Parallel realization of permutations over trees}, Discrete Math.
  \textbf{39} (1982), 211--214.

\bibitem{T83}
\bysame, \emph{Computation of boolean functions on networks of binary
  automata}, J. Comput. Systems Sci. \textbf{26} (1983), 269--277.

\bibitem{T85}
\bysame, \emph{Permutation factorization on star-connected networks of binary
  automata}, SIAM J. Algebraic Discrete Methods \textbf{6} (1985), 537--540.

\bibitem{T86}
\bysame, \emph{Computation on binary tree network}, Discrete Appl. Math.
  \textbf{14} (1986), 295--310.

\bibitem{T88}
\bysame, \emph{Computation on finite networks of automata}, Automata Networks
  (New York) (Christian Choffrut, ed.), Lecture Notes in Comput. Sci., vol.
  316, Springer-Verlag, 1988, pp.~53--67.

\end{thebibliography}

\end{document}